\documentclass[a4paper,12pt]{paper}
\pdfoutput=1

\usepackage{natbib}
\setcitestyle{authoryear,semicolon} 

\usepackage{amsmath,empheq}
\usepackage{amsthm}
\usepackage{framed}
\usepackage{graphicx}
\usepackage{subcaption}

\graphicspath{{./figures/}}    
\usepackage{algorithm}
\usepackage{algpseudocode}
\usepackage{array}
\usepackage[toc,page]{appendix}
\usepackage{amsfonts}
\usepackage{hyperref} 
\usepackage{setspace} 
\usepackage{authblk}

\usepackage{varwidth}
\usepackage[toc,page]{appendix}

\usepackage{xcolor}
\usepackage{mathtools}
\usepackage{framed}

\newtheorem{lemma}{Lemma}
\newtheorem{observation}{Observation}

\begin{document}
\begin{titlepage}
   \begin{center}
       \vspace*{1cm}
 		\large{
       \textbf{A New Paradigm for Identifying Reconciliation-Scenario Altering Mutations Conferring Environmental Adaptation\footnote{A conference version of this paper appeared in WABI-2019}}}

       \vspace{0.8cm}
              \vfill
	
       \large \textbf{Roni Zoller}\\
       \small{
       Computer Science Department, Ben Gurion University of the Negev, Israel, tel. +972-54-6472788, Email: ronizo@post.bgu.ac.il\\}
       \large \textbf{Meirav Zehavi*}\\
       \small{
       Computer Science Department, Ben Gurion University of the Negev, Israel, tel. +972-8-6428559, Email: meiravze@bgu.ac.il \\}
       \large \textbf{Michal Ziv-Ukelson*}\\
       \small{
       Computer Science Department, Ben Gurion University of the Negev, Israel, tel. +972-8-6428042, Email: michaluz@cs.bgu.ac.il\\}
       \vfill
 
       \vspace{0.8cm}
  
       *Corresponding authors.
 
   \end{center}
\end{titlepage}
\newpage

\begin{abstract}
An important goal in microbial computational genomics is to identify crucial events in the evolution of a gene that severely alter the duplication, loss and mobilization patterns of the gene within the genomes in which it disseminates. In this paper, we formalize this microbiological goal as a new pattern-matching problem in the domain of Gene tree and Species tree reconciliation, denoted  "Reconciliation-Scenario Altering Mutation (RSAM) Discovery". We propose an $O(m\cdot n\cdot k)$ time algorithm to solve this new problem, where $m$ and $n$ are the number of vertices of the input Gene tree and Species tree, respectively, and $k$ is a user-specified parameter that bounds from above the number of optimal solutions of interest. The algorithm first constructs a hypergraph representing the $k$ highest scoring reconciliation scenarios between the given Gene tree and Species tree, and then interrogates this hypergraph for subtrees matching a pre-specified RSAM Pattern. Our algorithm is optimal in the sense that the number of hypernodes in the hypergraph can be lower bounded by $\Omega(m\cdot n\cdot k)$.
We implement the new algorithm as a tool, called RSAM-finder, and demonstrate its application to -the identification of RSAMs in toxins and drug resistance elements across a dataset spanning hundreds of species.
\end{abstract}

\section{Introduction}\label{sec:introduction}

Prokaryotes can be found in the most diverse and severe ecological niches of the planet.  Adaptation of prokaryotes to new niches requires expanding their repertoire of protein families, via two evolutionary processes:  first, by selection of novel gene mutants carrying stable genetic alterations that confer adaptation, and second, by dissemination of an adaptively mutated gene. 
These two processes are correlated:  an adaptation-conferring mutation in a gene could accelerate its mobilization across bacterial lineages populating the corresponding environmental niche (\citet{poirel2009integron}),  
and vice-versa, the mobilization of a gene by transposable elements increases its chances to mutate or ``pick up" novel genomic context. 
Thus, an important research goal is to identify gene-level mutations that affect the spreading pattern of the mutated gene within and across the genomes harboring it.

For example, consider mutations conferring adaptation of bacteria to a human-pathogenesis environment. Here, a mutation to a resistance or virulence factor could enhance pathogenic adaptation, thus increasing the horizontal mobilization of the mutated gene within other human pathogens inhibiting this niche (\citet{poirel2009integron}).  In this case, we say that the mutation has a {\em causal association} with the observed dissemination pattern of the mutated gene (i.e.~the increased mobilization of the gene among pathogenic bacteria). Identifying such mutations could inform infectious disease monitoring and outbreak control, and assist in identifying~potential~drug~targets. 

\begin{figure}[t]
\begin{center}
\fbox{\includegraphics[scale=0.1]{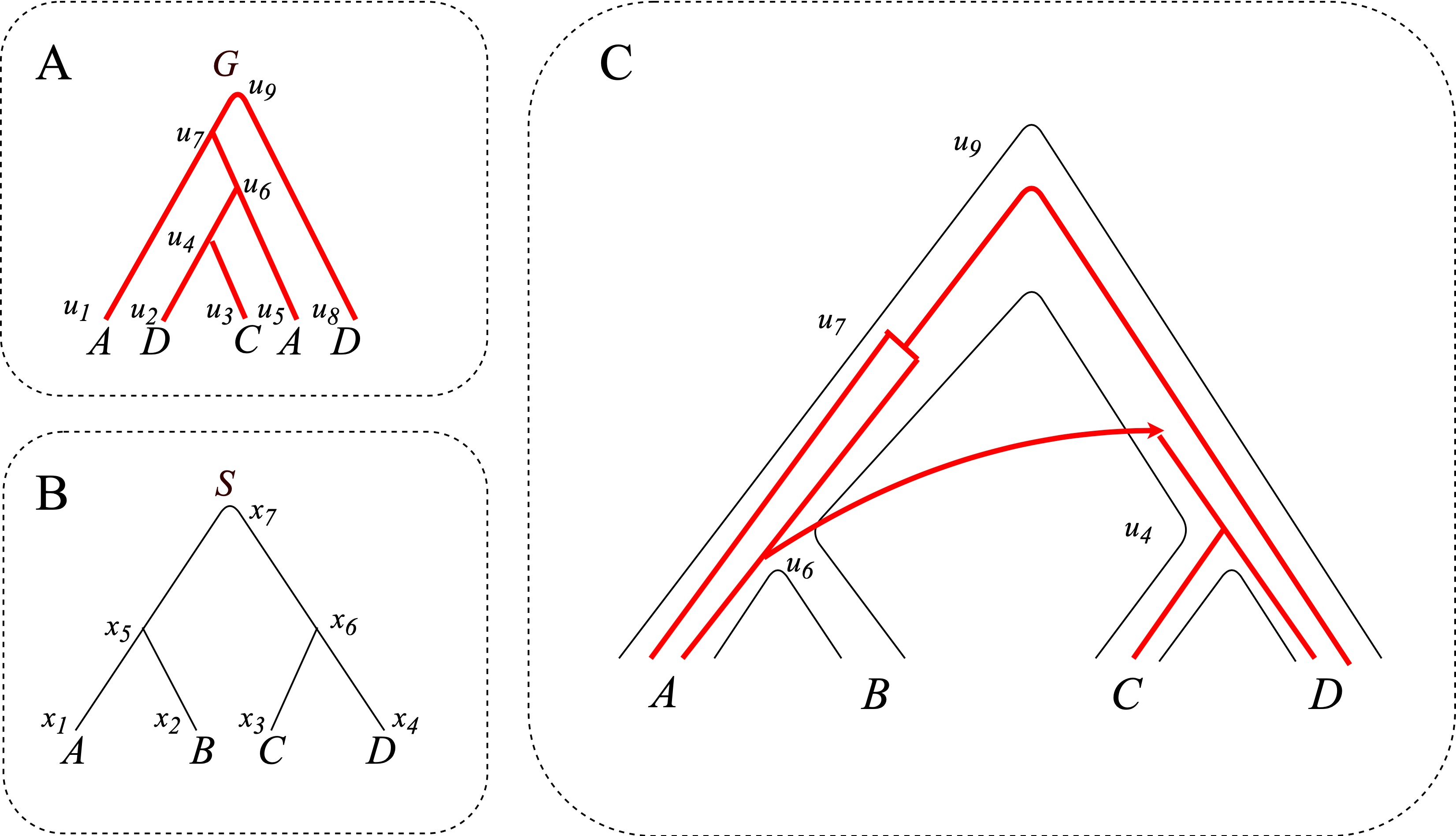}}
\end{center}
\caption{An example of a DLT scenario. (A) The Gene tree $G$. (B) The Species tree $S$. (C) A possible reconciliation scenario between $G$ and $S$.}
\label{fig:DLT}
\end{figure}

The co-evolution of genes and their host species is classically described by computing  the most parsimonious reconciliation scenario between a given Gene tree $G$ and the corresponding Species tree $S$, that is, a mapping of each vertex $u\in G$ to a vertex $x\in S$. Three major evolutionary processes, traditionally considered by reconciliation approaches, are horizontal gene transfer, gene duplication, and gene loss (\citet{tofigh2011simultaneous}). Each mapping of a vertex $u\in G$ to a vertex $x\in S$  is associated with one of these evolutionary events, and assigned a cost, accordingly (see Fig.~\ref{fig:DLT}). 
The optimization problem of computing a least-cost reconciliation between $G$ and $S$, where the total cost is computed as the sum of the costs assigned to each of the mappings, is denoted {\em Duplication-Transfer-Loss (DLT) Reconciliation}. (Previous works on this problem are reviewed in Section~\ref{sec:previous_works}~below.)

Motivated by examples such as the one given above, we formalize a new pattern-matching problem in the domain of DLT reconciliation. Given are a Gene tree $G$, a corresponding Species tree $S$,  a mapping $\sigma$ from the leaves of $G$ to the leaves of $S$, and (optional) an environmental annotation labeling the leaves of the input trees. Let $\mathcal{H}$ denote some data structure, to be defined later in the paper, that models the space of reconciliations between $G$ and $S$. 
A {\em DLT Reconciliation Scenario Pattern} denotes a mapping between a vertex $u\in G$ to a vertex $x\in S$, which obeys a set of user-defined specifications regarding the corresponding reconciliation event, the labels on the paired vertices, and other features associated with the mapping. Mappings between pairs of vertices ($u\in G$, $x\in S$) that abide by the requirements specified by $P$ are denoted {\em instances of $P$ in $\mathcal{H}$}. 
Given a pre-specified DLT Reconcilation Scenario pattern $P$ and a data structure $\mathcal{H}$ modeling the space of reconciliations between $G$ and $S$, a {\em Reconciliation Scenario Altering Mutation (RSAM) of $P$ in $\mathcal{H}$} is a vertex $v \in G$ representing a gene mutation with a putative causal association to instances of $P$ in $\mathcal{H}$. The {\em RSAM Discovery problem} is to identify RSAMs in $G$.

\begin{figure}[t]
\begin{center}
\fbox{\includegraphics[scale=0.115]{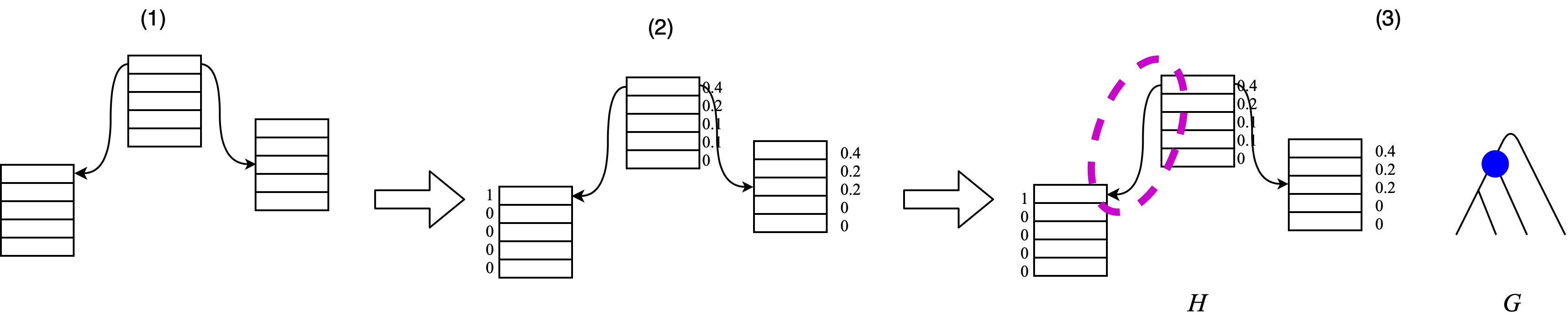}}
\end{center}
\caption{High-level overview of the RSAM-finder algorithm.}
\label{fig:flowchart}
\end{figure}

In what follows, we propose a three-stage solution to the RSAM Discovery problem defined above (illustrated in Fig.~\ref{fig:flowchart}).  The first stage constructs a hypergraph $\cal{H}$ that recursively aggregates all the $k$-best reconciliations of $G$ and $S$. Each supernode in $\cal{H}$ consists of $k$ hypernodes, where each hypernode represents a partial solution for the DLT-reconciliation problem. Our hypergraph-ensemble approach is based on a model proposed by \citet{patro2013predicting} for network evolution, where here we extend and adapt it to the DLT model. This hypergraph of $k$-best reconciliations, intended to provide some robustness to the noise typical of this data, will serve as the search-space for the pattern-matching stage.
The second stage of our proposed solution consists of assigning a probability to each partial solution, that is, to each hypernode of $\cal{H}$. 
Finally, in the third stage, instances of the sought RSAM-pattern $P$ are identified within $\cal{H}$, and  RSAM-ranking scores are assigned accordingly to the vertices of $G$. Based on these scores, vertices representing putative RSAMs are identified in $G$ and subjected to biological interpretation. 

%

The construction of $\cal{H}$, in the first stage, is the computational bottleneck of the RSAM-Discovery pipeline mentioned above. Here, we adapt the approach proposed by~\citet{bansal2012efficient} for the basic, 1-best variant of DLT reconciliation, extending it to an efficient $k$-best variant.  This yields an $O(m\cdot n\cdot k)$ time algorithm for the problem, where $m$ and $n$ are the number of vertices of the input Gene tree and Species tree, respectively, and $k$ is a user-specified parameter that bounds from above the number of optimal solutions of interest.  Our algorithm is optimal in the sense that the number of hypernodes in the hypergraph can be lower bounded by $\Omega(m\cdot n\cdot k)$.

We remark that a simpler, $O(mn (n+k)\log (n+ k))$ algorithm for hypergraph construction can be obtained by directly building upon the dynamic programming (DP) algorithm of~\citet{tofigh2011simultaneous} and employing the method of cube pruning by \citet{huang2005better} to handle lists of (partial) $k$-best solutions efficiently. Pseudocode for this "naive" algorithm can be found in \citet{supmat2019} Section 2. Just like~\citet{bansal2012efficient} shaved off one $n$ factor in the time complexity of the algorithm by~\citet{tofigh2011simultaneous}, so do we shave the $n$ factor in the term $(n+k)$ in the time complexity of the aforementioned naive algorithm. 
Surprisingly, we show that by relying on the improved DP algorithm, a further speed up is  achieved---namely, the usage of a queue becomes unnecessary and therefore the $\log(n+k)$ factor is {\em eliminated}. 

%


Our proposed solution to the problem defined in this paper is implemented as a tool called RSAM-finder, publicly available in \citet{zoller2019}. We assert the performance of RSAM-finder in large scale simulations, and exemplify its application to the identification of  RSAMs across a datasets spanning hundreds of species.

\subparagraph{Previous Related Works.} \label{sec:previous_works}

The DLT Reconciliation problem has been extensively studied. 
In particular, two main DLT variants have been considered: (1) the \textit{undated DLT-reconciliation}, where the species are undated, and (2) the \textit{fully-dated DLT-reconciliation}, where either each vertex in the Species (and Gene) tree is associated with an estimated date or the vertices of the Species  (and Gene) tree are associated with a total order, and any reconciliation must respect these dates or order (i.e.~an HT event can occur only between co-existing species).

In the acyclic version of these variants, there cannot exist two genes such that one is a descendant of the other, yet the descendant is mapped (in the Species tree $S$) to an ancestor of the other. \citet{tofigh2011simultaneous} showed that the acyclic undated version is NP-Hard. However, the acyclic dated version becomes polynomially solvable (\citet{libeskind2009computational}). 
 
\citet{tofigh2011simultaneous, tofigh2009using} and \citet{david2011rapid} studied a version of the undated (cyclic) problem that ignores losses and proposed an $O(mn^2)$ dynamic programming algorithm for it. They also gave a fixed-parameter tractable algorithm for enumerating all optimal solutions. The time complexity of the algorithm was improved to $O(mn)$ in \citet{tofigh2009using} (under a restricted model that ignores losses) and in \citet{bansal2012efficient} (which does not ignore losses).

It is well known that the biological data used as input to the DLT Reconciliation problem could be inaccurate, whether due to a sequencing problem, a problem in the reconstruction of $G$ or $S$ (\citet{bapteste2009prokaryotic}), or due to some other problem caused by noise. 
To overcome this problem, previous works try to examine more than one optimal solution, for example, see~\citet{donati2015eucalypt} and \citet{scornavacca2013representing}. A probabilistic method for exploring the space of optimal solutions was suggested in \citet{bansal2013reconciliation} and \citet{doyon2009space}, where the latter was improved in \citet{doyon2011efficient}. Additional studies considered a space of candidate co-optimal scenarios within special variants of the DLT problem, some of which employed special constraints to drive the search (\citet{stolzer2012inferring,to2015fast,merkle2010parameter,charleston1998jungles}). Although all of the previous works reviewed  in this paragraph compute a space of candidate reconciliation scenarios, none of these works considered the application of pattern matching on this space, as we do in this work.  

DLT Reconcilation variants where the reconcilation computation is guided by constraints derived from vertex-coloring information, were proposed in applications studying host-parasite co-evolution, such as \citet{berry2018geography}, where the vertex coloring (in both $G$ and $S$) represents the geographical area of residence. However, the applied  constraints were ``hard-wired'' to the specific problem addressed in that paper. In contrast, the approach proposed in this paper is more general, supporting a pattern-search that is guided by a user-defined pattern.  Our tool RSAM-finder provides the users with a query language able to express more robust patterns, according to the various applications where the pattern-search is to be employed.

\section{Preliminaries} \label{sec:definitions}

For a (binary) rooted tree $T$, let $L(T)$, $V(T)$, $I(T)$ and $E(T)$ denote the sets of leaves, vertices, internal vertices and edges, respectively, of $T$. Additionally, let $V(T)^\star$ denote the set of finite (ordered) vectors over $V(T)$, i.e.~$V(T)^\star=\{ (v_1,v_2,\dots, v_\ell)\mid v_i\in V(T)$ for all $i\in\{1,\ldots,\ell\},\ell\in \mathbb{N}\}$.  When $T$ is clear from context, let $V^\star=V(T)^\star$.
Throughout, we treat any (binary) rooted tree $T$ as a directed graph whose edges are directed from root to leaves.
Then, if $(u,v)\in E(T)$, we say that $v$ is a {\em child} of $u$, and $u$ is the {\em parent} of $v$. For $u,v\in V(T)$, the notation $v\leq _T u$ signifies that $v$ is a {\em descendant} of $u$ (alternatively, $u$ is an {\em ancestor} of $v$), i.e.~there is a directed path from $u$ to $v$ or $u=v$. We say that $v$ is a {\em proper descendent} (resp.~{\em proper ancestor}) of $u$ if $v\leq _T u$ (resp.~$v\geq _T u$) and $u\neq v$, denoted $<_T$ (resp.~$>_T$). When both $u\not \leq_T v$ and $v\not\leq _T u$, we say that $u$ and $v$ are {\em incomparable}. 
For any $u,v\in V(T)$, let $d_T(u,v)$ denote the number of edges in the (unique simple undirected) path between $u$ and $v$ in $T$. When $T$ is clear from context, we drop it from the notations $v\leq _T u$ and $d_T(u,v)$.
For any $u\in V(T)$, let $T_u$ denote the subtree of $T$ rooted in $u$ (then, $V(T_u)=\{v\in V(T)\mid v\leq u\}$). 

\subparagraph*{DLT Scenario.} A {\em DLT scenario} for two binary trees $G$ (the {\em Gene tree}) and $S$ (the {\em Species tree}) is a tuple $\left\langle  \sigma, \gamma, \Sigma, \Delta, \Theta, \Xi \right\rangle $ where $\sigma: L(G)\to L(S)$ is a mapping of the leaves of $G$ to the leaves of $S$, $\gamma:V(G)\to V(S)$ is a mapping of the vertices of $G$ to the vertices of $S$, and $(\Sigma, \Delta, \Theta)$ is a partition of $I(G)$ (the set of internal vertices of $G$) into three event classes: {\em Speciation} ($\Sigma$), {\em Duplication} ($\Delta$) and {\em Horizontal Transfer} ($\Theta$). The subset $\Xi \subseteq E(G)$ specifies which edges are involved in horizontal transfer events. Additionally, the following constraints should be satisfied.
\begin{enumerate}
	\item {\bf Consistency of $\sigma$ and $\gamma$.} For each leaf $u\in L(G),\ \gamma (u)=\sigma (u)$. {\em This constraint ensures that $\gamma$ respects $\sigma$---that is, each leaf of $G$ is mapped to the species where it is found.} \label{lst:line:DLT1}
	\item {\bf Consistency of $\gamma$ and ancestorship relations in $S$.} For each $u\in I(G)$ with children $v$ and $w$:
	\begin{enumerate}
		\item $\gamma(u)\not <_S \gamma (v)$ and $\gamma(u)\not <_S \gamma (w)$. {\em This constraint ensures that each of the two children (in $G$) of the gene $u$ is mapped by $\gamma$ to a species that is {\em not} a proper ancestor (in $S$) of the species to which the gene $u$ is mapped; thus, it can be either a descendant of $u$ or incomparable to $u$.} \label{lst:line:DLT2a}
		\item At least one of $\gamma(v)$ and $\gamma (w)$ is a descendant of $\gamma (u)$. {\em This constraint ensures that at least one of the two children (in $G$) of the gene $u$ is mapped by $\gamma$ to a species that is a descendant (in $S$) of the species to which the gene $u$ is mapped.} \label{lst:line:DLT2b}
	\end{enumerate}
	\item {\bf Identifying horizontal transfer edges.} For each edge $(u,v)\in E(G)$, it holds that $(u,v)\in \Xi$ if and only if $\gamma(u) \not\leq _S \gamma(v)$ and $\gamma(v) \not\leq _S \gamma(u)$. {\em This constraint identifies which edges are horizontal transfer edges---specifically, a horizontal transfer edge is an edge $(u,v)\in E(G)$ from a gene $u$ to a gene $v$ that are mapped to species $\gamma(u)$ and $\gamma(v)$ that are incomparable.} \label{lst:line:DLT3}
	\item {\bf Associating events with internal vertices.} For each $u\in I(G)$ with children $v, w$: \label{lst:line:DLT4}
	\begin{enumerate}
		\item {\bf Speciation.} $u\in \Sigma$ only if both {\em (i)} $\gamma (u)=\mathsf{lca} (\gamma (v), \gamma (w))$ and {\em (ii)} $\gamma(v)$ and $\gamma (w)$ are incomparable (i.e.~$\gamma(v) \not\leq _S \gamma(w)$ and $\gamma(w) \not\leq _S \gamma(v)$).
		\item {\bf Duplication.} $u\in \Delta$ only if $\gamma (u)\geq _S \mathsf{lca} (\gamma (v), \gamma (w))$.
		\item {\bf Horizontal transfer.} $u\in \Theta$ if and only if either {\em (i)} $(u,v)\in \Xi$ or {\em (ii)} $(u,w)\in \Xi$.
	\end{enumerate}
\end{enumerate}

Fig.~\ref{fig:DLT} demonstrates a DLT scenario. The species are written below the leaves of $S$. The (non-injective) mapping $\sigma: L(G)\rightarrow L(S)$ is implied by the labels of the leaves of $G$: $\sigma(u_1)=x_1;\ \sigma(u_2)=x_4;\ \sigma(u_3)=x_3;\ \sigma(u_5)=x_1;\ \sigma(u_8)=x_4$. In the DLT reconciliation of $G$, $S$ and $\sigma$ (Fig.~\ref{fig:DLT}.C), the tubes illustrate the edges of $S$, and each edge of $G$ is embedded inside the tube (edge of $S$) to which it is mapped by $\gamma$. Then, $\Sigma=\{ u_9,u_4\}$, $\Delta = \{ u_7\}$ and $\Theta=\{ u_6\}$. Moreover, $\Xi = \{ (u_6,u_4)\}$.

\subparagraph*{Losses.} Our definition of a loss event is based on the definition given by \citet{bansal2012efficient}.
Consider a Gene tree $G$, a Species tree $S$ and a corresponding DLT scenario $\alpha = \left\langle  \sigma, \gamma, \Sigma, \Delta, \Theta, \Xi  \right\rangle$. Let $u\in V(G)$ with children $v$ and $w$ (if they exist). Define $\mathsf{Loss}_\alpha (u)$ as the number of losses at $u$. Intuitively, the number of losses at a vertex $u$ is the number of ``skips" the gene made in the tree $S$ at the evolutionary event that $u$ represents. Formally,
\[\mathsf{Loss}_\alpha (u)=\begin{cases}
d_S (\gamma (u), \gamma (v))-1+d_S (\gamma (u), \gamma (w))-1 & u\in \Sigma\\
	d_S (\gamma (u), \gamma (v))+d_S (\gamma (u), \gamma (w))  & u\in \Delta \\
	d_S (\gamma (u), \gamma (v)) & (u,w)\in \Xi
\end{cases}
 \]
 
Recall that $d_S (u,v)$ is the distance between $u$ and $v$ in the tree Species $S$. The formula above determines that the number of losses in a vertex $u\in V(G)$ is based on the event that occurred in $u$. First, if $u\in \Sigma$ (i.e.~$u$ represents a speciation event), then the number of losses is the sum of the distances between $u$ and its two children in the Species tree (by the mapping $\gamma$) without counting the first step. If $u\in \Delta$ (i.e.~$u$ represents a duplication event), then the number of losses is the sum of the distances between $u$ and its two children in the Species tree.  If $(u,w)\in \Xi$ (i.e.~$u$ represents a horizontal transfer event that happened in the edge $(u,w)$), then the number of losses is the sum of the distances between $u$ and its other child (i.e.~$v$) in the Species tree.

\subparagraph*{Costs.} Let $c_\Sigma,\ c_\Delta,\ c_\Theta$ and $c_{\mathsf{loss}}$ denote the costs of a speciation event, a duplication event, a horizontal transfer event and a loss event, respectively.
Let $\mathsf{Loss}_\alpha = \sum _{u\in V(G)}\mathsf{Loss}_\alpha (u)$.
Let $|\Sigma|\cdot c_\Sigma+|\Delta|\cdot c_\Delta  + |\Theta|\cdot c_\Theta + \mathsf{Loss}_\alpha \cdot c_{\mathsf{loss}}$ be the reconciliation cost of $\alpha$. When seeking a ``best'' DLT scenario, the goal is to find one that minimizes this cost.

\section{Hypergraph of \texorpdfstring{$k$}{Lg}-Best Scenarios} \label{sec:hypergraph-def}
\begin{figure}
\begin{center}
   \fbox{\includegraphics[scale=0.08]{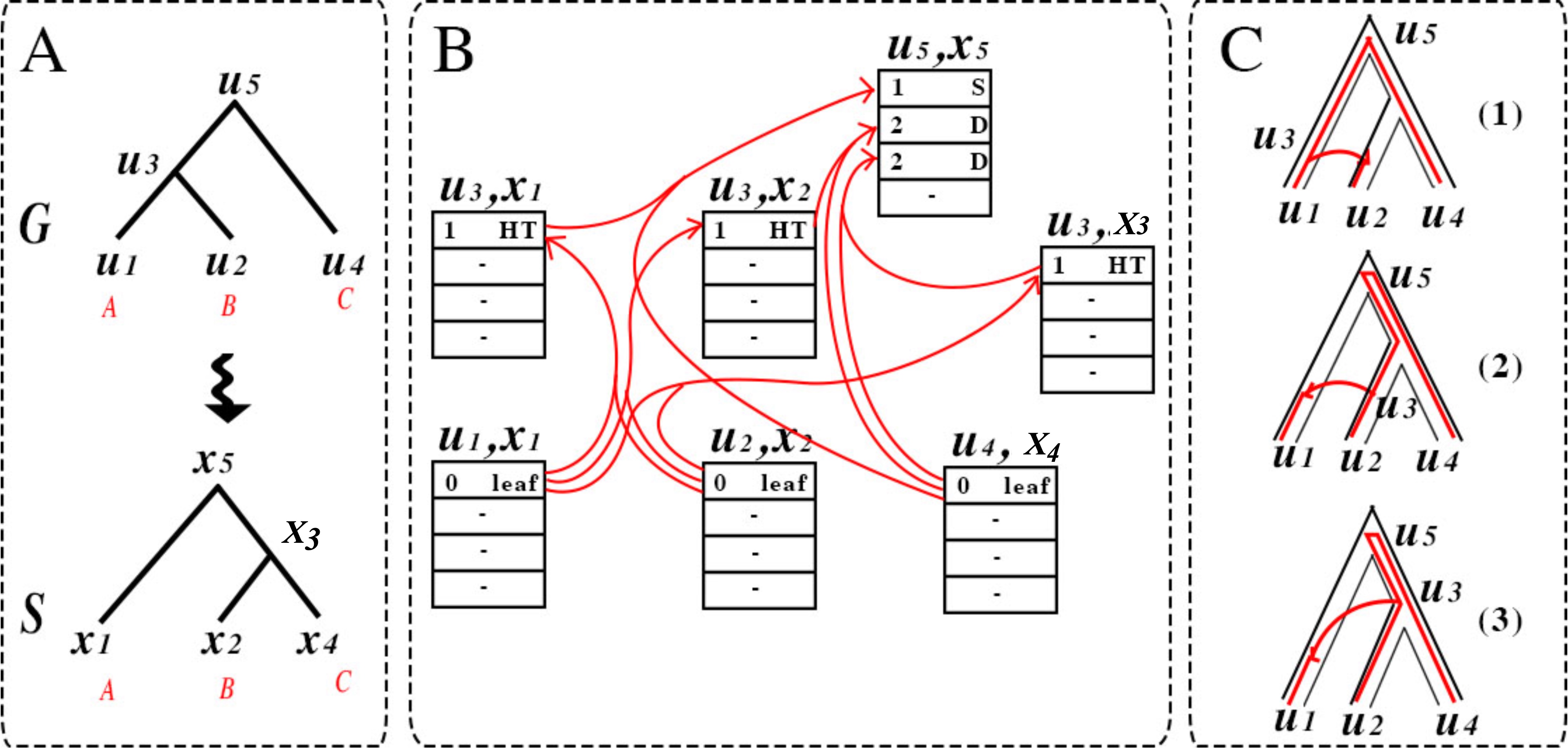}}
   \caption{Various aspects of the problem addressed in this paper. (A) The input trees $G$ and $S$. (B) An example of a hypergraph constructed based on the input trees and parameter $k$. (C) Three top-scoring DLT-reconciliations for the input.}
	\label{fig:hypergraph}
    \label{fig:fig1}
    \end{center}
\end{figure}

To represent {\em $k$-best solutions},\footnote{That is, $k$ DLT scenarios of the highest score(s), where ties (if any exist) are broken arbitrarily.} we use a directed hypergraph denoted by  $\mathcal {H}$ based on the notation in~\citet{huang2005better}. The hypergraph is a tuple $\mathcal{H}=\left\langle V,E\right\rangle$, where $V$ is a finite set of hypernodes, and $E$ is a finite set of (directed) hyperedges defined as follows. Each $e\in E$ is a pair $\left\langle T(e),h(e)\right\rangle $, where $h(e)\in V$ is the {\em head} of $e$, and $T(e)\in V^*$ (i.e.~$T(e)$ is a vector of vertices in $V$) is its {\em tail}. In our settings, $|T(e)|=2$ for every $e\in E$. In what follows, we define the hypernodes and hyperedges of $\mathcal{H}$ with respect to our problem. To exemplify this, we refer the reader to Fig.~\ref{fig:hypergraph}. In part B of this figure, the hypernode $(u_5,x_5,1)$ is annotated with score $1$ and event 'S'. To extract the best solution from the hypergraph, we begin with the first (i.e.~top-ranking) slot in the root of the hypergraph (which is $(u_5,x_5,1)$ in the figure), and then follow the incoming hyperedges in top-down order. In the figure, the best solution is solution (1) in part C of the figure. To extract it from the hypergraph in part B of the figure, we map $u_5$ to $x_5$ with a cost of 1 and a speciation event. Then, by first following the hyperedges incoming to $(u_5,x_5,1)$, we derive the mapping of $u_3$ to $x_1$, and of $u_4$ to $x_3$. Finally, by following the hyperedges incoming to $(u_3,x_1,1)$, we also derive the mapping of $u_1$ to $x_1$, and of $u_2$ to $x_2$.

The second best solution (solution (2) in part C of the figure) is extracted in the same manner---now, we start with the hypernode $(u_5,x_5,2)$ rather than $(u_5,x_5,1)$, and again follow incoming hyperedges in a top-down order until we reach the leaves. Similarly, we can extract all three non-nil solutions among the $4$-best solutions (illustrated in part C). As before, the outer tubes illustrate the edges of $S$, and the edges of $G$ are embedded inside based on the reconciliation.

\begin{itemize}
\item{{\bf Hypernodes.} For every vertex $u$ in $G$, a vertex $x$ in $S$ and an integer $i\in\{1,\ldots,k\}$, we have a hypernode $(u,x,i)$ in $\cal H$. Such a hypernode $(u,x,i)$ is associated with the $i^{\mathrm{th}}$ best (where ties are broken arbitrarily) solution mapping the subtree of $G$ rooted in $u$ to the subtree of $S$ rooted in $x$ that is a DLT scenario. In addition, for every integer $i\in\{1,\ldots,k\}$  we have a hypernode $(root,i)$ in the hypergraph $\cal{H}$. Such a hypernode $(root,i)$  is associated with the $i^{\mathrm{th}}$ best solution of mapping $G$ (entirely) to any subtree of $S$. Each hypernode $(u,x,i)$ has a score $\mathsf{c}(u,x,i)$, and each hypernode $(root,i)$ has a score $\mathsf{c}(root,i)$.
Moreover, each hypernode $(u,x,i)$ is associated with the  event corresponding to the mapping of $u$ and $x$ in the DLT scenario of $(u,x,i)$ (speciation, duplication or horizontal transfer), denoted $\mathsf{event}(u,x,i)$.}

\item{{\bf Supernodes.} For any vertex $u\in V(G)$ and vertex $x\in V(S)$, we define the {\em supernode} $(u,x)$ as the list $\left \{(u,x,i):1\leq i\leq k \right \}$ (i.e.~$(u,x)$ is the set of $k$ hypernodes corresponding to the mapping of the subtree of $G$ rooted in $u$ to the subtree of $G$ rooted in $x$). This notation will simplify our presentation.}

\item{{\bf Hyperedges.} We remind the reader that each hypernode $(u,x,i)\in V(\cal H)$ describes a DLT scenario. Each hypernode has exactly one incoming hyperedge, but it can have multiple outgoing hyperedges. In particular, for each hypernode $(u,x,i)\in V(\cal{H})$, the (only) incoming hyperedge $e=\left\langle T(e),h(e)\right\rangle=\left\langle [(v,y,j),(w,z,r)],(u,x,i)\right\rangle$ describes the mapping of the subtrees of the children  of $u$, namely, $v$ and $w$, in the scenario of $(u,x,i)$; here, the subtree of $v$ is mapped to the subtree of $y$ as in the scenario of $(v,y,j)$, and the subtree of $w$ is mapped to the subtree of $z$ as in the scenario of $(w,z,r)$.}
\end{itemize}

\section{Framework and Algorithms}

In this section, we elaborate on each of the three stages of the workflow in Section \ref{sec:introduction}.

\subsection{Stage 1: Hypergraph Construction}\label{sec:hypergraphCons}

The first stage of our framework is to construct the hypergraph described in Section~\ref{sec:hypergraph-def}. 
To this end, we develop an efficient algorithm that runs in time $O(m\cdot n\cdot k)$ and requires $O(m\cdot n\cdot k)$ space.

\subparagraph*{An Overview of the Algorithm.}
We iterate over all $u\in V(G)$ in postorder, as well as over all $x\in V(S)$ in postorder. (However, as explained immediately, 
when we consider a vertex $u\in V(G)$, after iterating over all vertices $x\in V(S)$ in postorder, we also iterate over all vertices $x\in V(S)$ in preorder.) In each iteration, corresponding to a pair $(u,x)$, we construct three lists: $p_\Sigma$ (speciation), $p_\Delta$ (duplication) and $p_\Theta$ (horizontal transfer). 
Specifically, $p_\Sigma$ should be a list of $k$-best solutions that are DLT scenarios where the subtree of $G$ rooted in $u$ is mapped to the subtree of $S$ rooted in $x$ under the restriction that the event corresponding to matching $u$ and $x$ is speciation. The meaning of the lists $p_\Delta$ and $p_\Theta$ is similar, where the restriction of speciation is replaced by duplication or horizontal transfer, respectively. 
Having these three lists suffices to construct the hypernode $(u,x)$.

To avoid  repetitive computation, we maintain three additional lists: $\mathsf{subtree}$, $\mathsf{subtreeLoss}$ and $\mathsf{incomp}$. Intuitively, $\mathsf{subtreeLoss}(u,x,i)$ represents the $i^{th}$ best cost of reconciliation of the tree rooted in $u$, such that $u$ may be mapped to any $y\leq x$ with a additional cost of one loss per edge in the path from $x$ to $y$, and $\mathsf{incomp}(u,x,i)$ represents the $i^{\mathrm{th}}$ best cost of a reconciliation of the subtree of $G$ rooted in $u$ with some subtree of $S$ whose root is a vertex $y$ incomparable to $x$. $\mathsf{subtree}$ is used in order to efficiently compute $\mathsf{incomp}$.
The notations $\mathsf{subtreeLoss}(u,x)$ , $\mathsf{subtree}(u,x)$ and $\mathsf{incomp}(u,x)$ refer to the lists of the $k$-best scores $\{\mathsf{subtreeLoss}(u,x,i)\}_{i=1}^k$, $\{\mathsf{subtree}(u,x,i)\}_{i=1}^k$ and $\{\mathsf{incomp}(u,x,i)\}_{i=1}^k$, respectively, similarly to our usage of the notation of a supernode.

The efficient computation of $p_\Sigma$, $p_\Delta$  and $p_\Theta$, along with the maintenance of $\mathsf{subtreeLoss}$, $\mathsf{subtree}$ and $\mathsf{incomp}$ themselves, is highly non-trivial.
On a high-level, we first initialize all five lists to contain only costs of $\infty$; then, still in the initialization phase, we add hypernodes that match between leaves of $G$ and $S$ in accordance with $\sigma$ and update $\mathsf{subtreeLoss}$ and $\mathsf{subtree}$ consequently. After the initialization, the main computation considers each $u\in V(G)$ in postorder, and performs two steps. In the first step, we consider each $x\in V(S)$ in postorder. Then, for each $i\in\{1,\ldots,k\}$, we compute $p_\Sigma(u,x,i)$, $p_\Delta(u,x,i)$ and $p_\Theta(u,x,i)$ based on somewhat involved recursive formulas. Afterwards, we construct the surpernode $(u,x)$, as well as compute the lists $\mathsf{subtreeLoss}(u,x)$ and $\mathsf{subtree}(u,x)$. In the second step, we consider each $x\in I(S)$ with children $y$ and $z$ in preorder, and compute the lists $\mathsf{incomp}(u,y)$ and $\mathsf{incomp}(u,z)$.

Having constructed all hypernodes of the form $(u,x,i)$ along with their ingoing hyperedges, it is trivial to construct the hypernodes of the form $(root,i)$ and their ingoing edges. 

\begin{figure}[htbp]
\begin{center}\fbox{\includegraphics[scale=1]{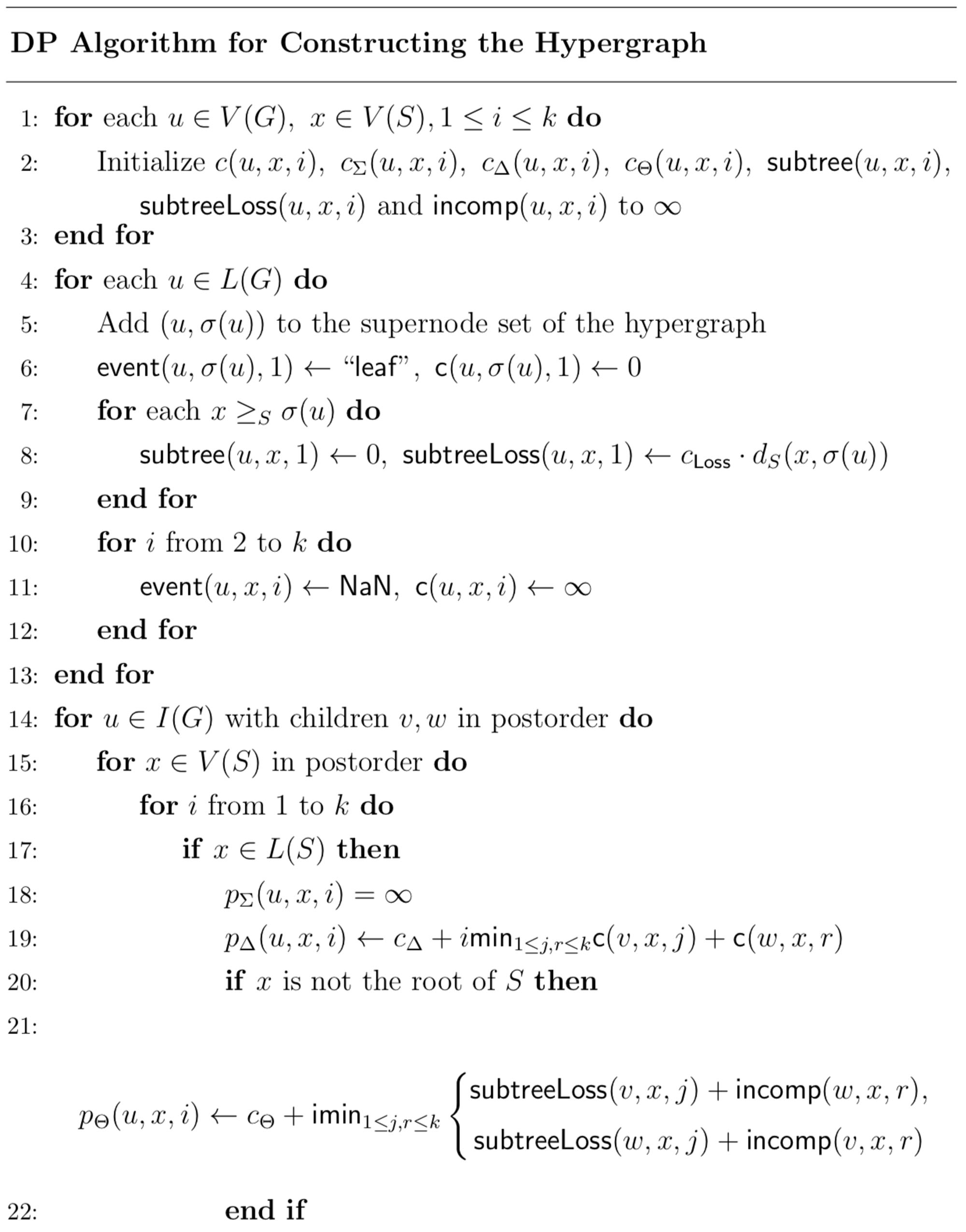}}  \caption{Psudocode of the algorithm (first part). The pseudocode is continued in Fig. 5. \label{fig:psaudo_1}}
\end{center}
\end{figure}

\begin{figure}[htbp]
\begin{center}\fbox{\includegraphics[scale=1]{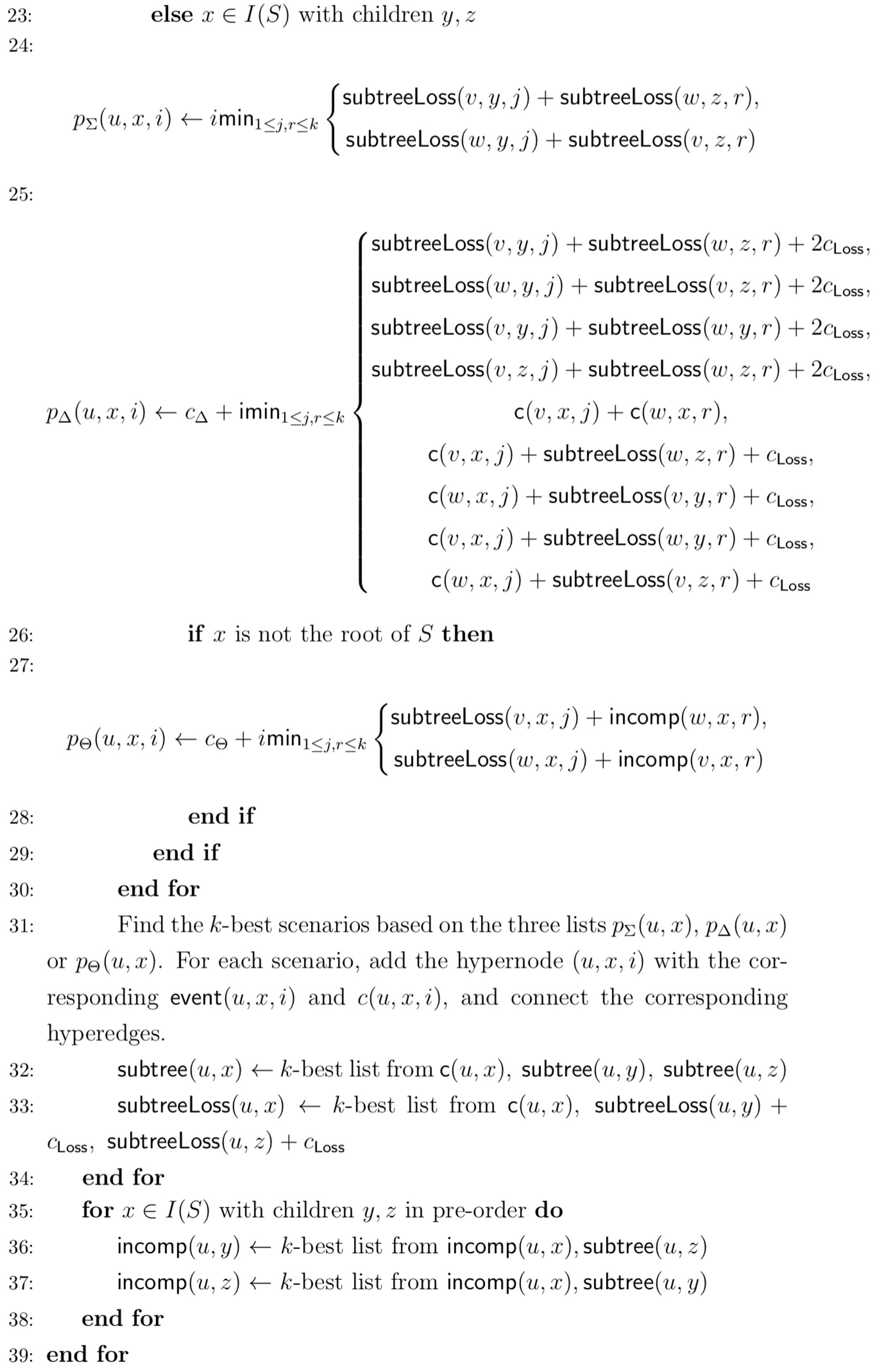}}  \caption{Psudocode of the algorithm (second part). This figure continues the pseudocode given in Fig. 4.\label{fig:psaudo_2}}
\end{center}
\end{figure}

\subparagraph*{Psudocode.}\label{sec:effi} The psudocode is given in Fig. \ref{fig:psaudo_1} and \ref{fig:psaudo_2}.
We use the notation $i \mathsf{min}$, defined as follows:
Let $X$ and $Y$ be sets, and consider a function $f:X\to Y$, and an index $i\in \{ 1,\dots ,|X|\}$.
Then, $i \mathsf{min}_{x'\in X} f(x')\overset{\Delta}{=}f(x)$ where $x$ is an element in $X$ such that there are exactly $i$ elements  $x'\in X$ satisfying $f(x')\leq f(x)$. In case $f$ is not an injective function, hence there are multiple choices for $x$, we break ties arbitrarily.

\noindent We proceed with a few clarifications of the pseudocode.

{\bf Initialization: Lines 1-13.} We initialize all lists to contain only scores of $\infty$ (lines 1-3). Then, the lists  associated with a matching between leaves that comply with $\sigma$---that is, supernodes of the form $(u,\sigma(u))$ for some $u\in V(G)$---are inserted into the hypergraphs, and their topmost items are updated with a leaf event, cost 0, and $\mathsf{subtreeLoss}$ and $\mathsf{subtree}$ 0 (because the cost of the best solutions mapping a gene to its species is~0).

{\bf Division into First and Second Phases: Lines 14-39.} For each vertex $u\in I(G)$ in postorder (line 14), we have two phases, on which we elaborate below. In the first phase (lines 15-34), we consider each vertex $x\in V(S)$ in postorder and perform most computations, and in the second phase (lines 35-38) we consider each vertex $x\in V(S)$ in postorder and compute the lists of $\mathsf{incomp}$.

{\bf Recursive Formulas for $p_\Sigma$, $p_\Delta$  and $p_\Theta$: Lines 16-30.} In this part of the first phase, we find the $k$-best costs for mapping the subtree of $G$ rooted in $u$ to the subtree of $S$ rooted in $x$ for each possible event (speciation, duplication or horizontal transfer), based on  computations done in previous iterations or the initialization. The recursive formulas for these computations are directly given in the pseudocode. 



{\bf Updating $\mathsf{c},\ \mathsf{subtreeLoss}$ and $\mathsf{subtree}$ in First Phase: Lines 31-32.} First, in line 31, we immediately find $k$-best costs for mapping the subtree of $u$ to the subtree of $x$ (i.e.~we compute $\mathsf{c}(u,x)$) by selecting $k$-best costs from the list that is the combination of $p_\Sigma(u,x)$, $p_\Delta(u,x)$  and $p_\Theta(u,x)$. Notice that in this line, we also add the appropriate hypernodes and hyperedges to the hypergraph. $\mathsf{event}(u,x,i)$ is defined by the source list ($p_\Sigma(u,x)$, $p_\Delta(u,x)$  or $p_\Theta(u,x)$) it came from. As before, if the combined list is shorter than $k$, we add hypernodes with $\mathsf{event}=\mathsf{Nan}$ and $\mathsf{cost}=\infty$. Secondly, in lines 32-33, we find $k$-best costs for mapping the subtree of $u$ to the subtree of some vertex $x'$ in the subtree of $x$, with and without loss events (i.e.~we compute $\mathsf{subtreeLoss}(u,x)$ and $\mathsf{subtree}(u,x)$) by selecting $k$-best costs from the combination of pre-calculated lists. 

{\bf Updating $\mathsf{incomp}$ in Second Phase: Lines 35-38.} To compute the lists of the form $\mathsf{incomp}(u,\cdot)$, in the second phase we iterate over all vertices $x\in I(S)$ with children $y$ and $z$ in {\em preorder}. We note that now the traversal of $S$ is in preorder rather than postorder because the computation of a list $\mathsf{incomp}(u,a)$ for a vertex $a\in V(S)$ that is not the root of $S$ relies on having already computed the list $\mathsf{incomp}(u,b)$ where $b$ is the parent of $a$ in $S$. Specifically, for a vertex $x\in I(S)$ with children $y$ and $z$, we compute the list $\mathsf{incomp}(u,y)$ by selecting $k$-best costs from the list that is the combination of $\mathsf{incomp}(u,x)$ and $\mathsf{subtree}(u,z)$, and symmetrically for $\mathsf{incomp}(u,z)$ (swapping the roles of $y$ and~$z$).

\begin{lemma}\label{lem:effCorrectness}
Given an instance $(G,S,\sigma)$ of the DLT problem and a positive integer $k$, the algorithm correctly constructs a hypergraph $\cal H$ that represents $k$-best solutions for $(G,S,\sigma)$.
\end{lemma}

\begin{proof}[Proof of Lemma \ref{lem:effCorrectness}.] 
We prove that for every pair of vertices $u\in V(G)$ and $x\in V(S)$, and every index $i\in\{1,\ldots,k\}$, if there exists an $i^{\mathrm{th}}$ best DLT scenario mapping the subtree of $G$ rooted in $u$ to the subtree of $S$ rooted in $x$, then the hypernode $(u,x,i)$ is inserted into the hypergraph $\cal H$ under construction with association to this scenario. In this lemma, the proof of this claim is done in conjunction with the proof that for every pair of vertices $u\in V(G)$ and $x\in V(S)$, and every index $i\in\{1,\ldots,k\}$, the following equalities hold.

\begin{itemize}
\item $\mathsf{subtreeLoss}(u,x,i)$ is the $i^{\mathrm{th}}$ best cost of a DLT scenario mapping the subtree of $G$ rooted in $u$ to some subtree of $S$ whose root is a vertex $y$ that is a descendant of $x$, with additional cost of one loss per each edge in the path from $x$ to $y$.
\item $\mathsf{subtree}(u,x,i)$ is the $i^{\mathrm{th}}$ best cost of a DLT scenario mapping the subtree of $G$ rooted in $u$ to some subtree of $S$ whose root is a vertex $y$ that is a descendant of $x$.
\item $\mathsf{incomp}(u,x,i)$ is the $i^{\mathrm{th}}$ best cost of a DLT scenario mapping the subtree of $G$ rooted in $u$ to some subtree of $S$ whose root is a vertex $y$ that is incomparable to $x$.
\end{itemize}

The proof is by induction on the order of computation. \\In particular, $\mathsf{table}_1(u_1,x_1)<\mathsf{table}_2(u_2,x_2)$ where $\mathsf{table}_1,\mathsf{table}_2\in\{\mathsf{c},\\ \mathsf{subtreeLoss},\mathsf{subtree},\mathsf{incomp}\}$ if one of the following conditions holds:
\begin{itemize}
\item $u_1$ is visited before $u_2$ in the postorder traversal of $G$.
\item $u_1=u_2$, $\mathsf{table}_1,\mathsf{table}_2\in\{\mathsf{c},\mathsf{subtreeLoss},\mathsf{subtree}\}$, and $x_1$ is visited before $x_2$ in the {\em postorder} traversal~of~$S$.
\item $u_1=u_2$, $x_1=x_2$, $\mathsf{table}_1=\mathsf{c}$ and $\mathsf{table}_2\in \{\mathsf{subtreeLoss},\mathsf{subtree}\}$.
\item $u_1=u_2$, $\mathsf{table}_1\in \{\mathsf{subtreeLoss},\mathsf{subtree}\}$ and $\mathsf{table}_2=\mathsf{incomp}$.
\item $u_1=u_2$, $\mathsf{table}_1=\mathsf{table}_2=\mathsf{incomp}$, and $x_1$ is visited before $x_2$ in the {\em preorder} traversal of $S$.
\end{itemize}

The basis of the induction comprises of the computation of hypernodes of the form $(u,x,i)$ where $u\in L(G)$. To prove its correctness, consider such a hypernode $(u,x,i)$. If $i=1$ and $x=\sigma(u)$, then the algorithm inserts the hypernode $(u,x,1)$, assigning it a score of $0$, and setting the remaining fields as follows: $\mathsf{subtreeLoss}=0$ , $\mathsf{subtree}=0$ and $\mathsf{event=leaf}$; else, if $i=1$ and $x\geq_S \sigma(u)$, then the algorithm inserts the hypernode $(u,x,1)$, assigning it a score of $0$ and setting the remaining fields as follows: $\mathsf{subtreeLoss}=c_{\mathsf{Loss}}\cdot d_S(x,\sigma(u))$, $\mathsf{subtree}=0$ and $\mathsf{event=leaf}$; otherwise, the algorithm does not insert the hypernode---more precisely, it inserts a place-holder (whose event is NaN) with score $\infty$ and $\mathsf{subtree}$ value $\infty$. In both cases, $\mathsf{incomp}$ value remains $\infty$ as in its creation. The correctness of these operations directly follows from the definitions of $\mathsf{subtreeLoss}$, $\mathsf{subtree}$ and $\mathsf{incomp}$, and the fact that the only possible DLT scenario in this case maps $u$ to an ancestor of $\sigma(u)$, and the score of this match is $0$ in case losses are not counted, or with the additional loss costs otherwise.

For the inductive step, we consider some pair of vertices $u\in I(G)$ and $x\in V(S)$ along with a table $\mathsf{table}\in\{\mathsf{c},\mathsf{subtreeLoss},\mathsf{subtree},\mathsf{incomp}\}$, and prove that the values in $\mathsf{table}$ of the supernode $(u,x)$ are computed correctly.  For the inductive assumption,  suppose that for every triple $(\mathsf{table}',u',x')$ ordered before $(\mathsf{table},u,x)$, the values in $\mathsf{table}'$ of $(u',x')$ have already been computed~correctly. Here we provide a proof for $\mathsf{table}=\mathsf{c}$ and $\mathsf{table}=\mathsf{incomp}$. The full proof can be found in Section~3 of \cite{supmat2019}.

First, consider the case where $\mathsf{table}=\mathsf{incomp}$. By the pseudocode, if $x$ is the root of $S$, then $\mathsf{incomp}(u,x)$ does not contain any item (having score different from $\infty$) as in its creation, which is correct because in this case, there exists no vertex incomparable to $x$ and hence we cannot map one of the children of $u$ as required in the definition of the DLT scenarios that correspond to $\mathsf{incomp}(u,x)$.  Therefore, now suppose that $x$ is not the root of $S$, and let $p$ denote the parent of $x$ in $S$, and $s$ denote the sibling of $x$ in $S$ (i.e.~the other child of $p$ in $S$). Then, by the pseudocode, $\mathsf{incomp}(u,x)$ consists of the $k$-best scores from the lists $\mathsf{incomp}(u,p)$ and $\mathsf{subtree}(u,s)$. Observe that these two lists have already been computed. Thus, by the inductive hypothesis, $\mathsf{incomp}(u,p)$ consists of the scores of the $k$-best DLT scenarios mapping the subtree of $G$ rooted in $u$ to the subtree of $S$ rooted in some vertex incomparable to $p$, and $\mathsf{subtree}(u,s)$ consists of the scores of the $k$-best DLT scenarios mapping the subtree of $G$ rooted in $u$ to the subtree of $S$ rooted in some descendant of $s$. Notice that a DLT scenario maps the subtree of $G$ rooted in $u$ to a subtree of $S$ rooted in some vertex incomparable to $x$ if and only if it is a DLT scenario that maps the subtree of $G$ rooted in $u$ to one of the following subtrees: {\em (i)} a subtree of $S$ rooted in some vertex incomparable to $p$; {\em (ii)} a subtree of $S$ rooted in some descendant of $s$. Thus, it follows that $\mathsf{incomp}(u,x)$ is computed correctly.

Second, consider the case where $\mathsf{table}=\mathsf{c}$. By line 31 of the pseudocode, $\mathsf{c}(u,x)$ consists of the $k$-best scores from the lists $p_\Sigma(u,x)$, $p_\Delta(u,x)$  and $p_\Theta(u,x)$.  Thus, to prove the correctness of the computation of $\mathsf{c}(u,x)$, it suffices to prove that the following statement holds: $p_\Sigma(u,x)$, $p_\Delta(u,x)$  and $p_\Theta(u,x)$ consist of $k$-best DLT scenarios mapping the subtree of $G$ rooted in $u$ to the subtree of $S$ rooted in $x$ under the constraint that the event corresponding to the matching of $u$ and $x$ is speciation, duplication and horizontal transfer, respectively.

Towards the proof of the statement, consider the list $p_\Sigma(u,x)$. If $x\in L(S)$, then because $u\in I(G)$, there does not exist a DLT scenario mapping the subtree of $G$ rooted in $u$ to the subtree of $S$ rooted in $x$ under the constraint that the event corresponding to the matching of $u$ and $x$ is speciation, and hence the assignment of $\infty$ to every element $p_\Sigma(u,x,i)$ of the list is correct. Now, suppose that $x\in I(S)$. Then, by the pseudocode, $p_\Sigma(u,x)$ consists of the $k$-best scores present in the following multisets: 
\begin{itemize}
\item $\{\mathsf{subtreeLoss}(v,y,j)+\mathsf{subtreeLoss}(w,z,r)~|~1\leq j,r\leq k\}$ and 
\item $\{\mathsf{subtreeLoss}(w,y,j)+\mathsf{subtreeLoss}(v,z,r)~|~1\leq j,r\leq k\}$. 
\end{itemize}
Observe that the lists $\mathsf{subtreeLoss}(v,y),\mathsf{subtreeLoss}(w,z)$, $\mathsf{subtreeLoss}(w,y)$ and $\mathsf{subtreeLoss}(v,z)$ have already been computed.  By the definition of $\mathsf{Loss}_{\alpha}(u)$ when the event occurred in $u$ is speciation,  $\mathsf{Loss}_{\alpha}(u)=|d_S (x, \gamma (v))-1|+|d_S (x, \gamma (w))-1|=d_S (y, \gamma (w))+d_S (z, \gamma (v))$ in case $\gamma(w)\leq_S y$, and $\mathsf{Loss}_{\alpha}(u)=|d_S (x, \gamma (v))-1|+|d_S (x, \gamma (w))-1|=d_S (z, \gamma (w))+d_S (y, \gamma (v))$ otherwise.  Thus, by the inductive hypothesis, $\mathsf{subtreeLoss}(v,y)$ (resp., \\ $\mathsf{subtreeLoss}(w,z)$, $\mathsf{subtreeLoss}(w,y)$ and $\mathsf{subtreeLoss}(v,z)$) consists of the scores of the $k$-best DLT scenarios mapping the subtree of $G$ rooted in $v$ (resp., $w,w$ and $v$) to a subtree of $S$ rooted in some descendant of $y$ (resp., $z,y$ and $z$), with additional loss cost for each edge in the path from $y$ to $\gamma(v)$ (resp., $\gamma(w)$, $\gamma(w)$ and $\gamma(v)$). Notice that a DLT scenario maps the subtree of $G$ rooted in $u$ to a subtree of $S$ rooted in $x$ under the constraint that the event corresponding to the matching of $u$ and $x$ is speciation if and only if it is a DLT scenario that matches $u$ and $x$, maps the subtree of $G$ rooted in $v$ to a subtree of $S$ rooted in a descendant of one child ($y$ or $z$) of $x$, and the subtree of $G$ rooted in $w$ to a subtree of $S$ rooted in a descendant of the other child of $x$. Thus, it follows that $p_\Sigma(u,x)$ is computed correctly.

Now, consider the list $p_\Delta(u,x)$. In case $x\in I(S)$, let $y$ and $z$ denote its children. By the pseudocode, $p_\Delta(u,x)$ consists of the $k$-best scores obtained by adding $c_{\Delta}$ to the costs present in the following multisets, where only the first one is relevant in case $x\in L(S)$: 
\begin{itemize}
\item $\{\mathsf{c}(v,x,j)+ \mathsf{c}(w,x,r)~|~1\leq j,r\leq k\}$, 
\item $\{\mathsf{c}(v,x,j)+\mathsf{subtreeLoss}(w,z,r)~|~1\leq j,r\leq k\}+c_{\mathsf{Loss}}$,
\item $\{\mathsf{c}(w,x,j)+\mathsf{subtreeLoss}(v,y,r)~|~1\leq j,r\leq k\}+c_{\mathsf{Loss}}$, 
\item $\{\mathsf{c}(v,x,j)+\mathsf{subtreeLoss}(w,y,r)~|~1\leq j,r\leq k\}+c_{\mathsf{Loss}}$,
\item $\{\mathsf{c}(w,x,j)+\mathsf{subtreeLoss}(v,z,r)~|~1\leq j,r\leq k\}+c_{\mathsf{Loss}}$,
\item $\{\mathsf{subtreeLoss}(v,y,j)+\mathsf{subtreeLoss}(w,z,r)~|~1\leq j,r\leq k\}+2c_{\mathsf{Loss}}$, 
\item $\{\mathsf{subtreeLoss}(w,y,j)+\mathsf{subtreeLoss}(v,z,r)~|~1\leq j,r\leq k\}+2c_{\mathsf{Loss}}$, 
\item $\{\mathsf{subtreeLoss}(v,y,j)+\mathsf{subtreeLoss}(w,y,r)~|~1\leq j,r\leq k\}+2c_{\mathsf{Loss}}$ and 
\item $\{\mathsf{subtreeLoss}(v,z,j)+\mathsf{subtreeLoss}(w,z,r)~|~1\leq j,r\leq k\}+2c_{\mathsf{Loss}}$. 
 \end{itemize}
Observe that the lists $\mathsf{c}(v,x), \mathsf{c}(w,x),\mathsf{subtreeLoss}(w,y)$, $\mathsf{subtreeLoss}(v,z)$, \\$\mathsf{subtreeLoss}(w,z)$ and $\mathsf{subtreeLoss}(v,y)$ have already been computed. By the definition of $\mathsf{Loss}_\alpha(u)$ when the event occurred in $u$ is duplication, $\mathsf{Loss}_\alpha(u)=d_S (x, \gamma (v))+d_S (x, \gamma (w))$. If $v$ (resp. $w$) is mapped to $x$ and $w$ (resp. $v$) is mapped to a subtree of $S$  rooted in $y$ or $z$, it holds that $\mathsf{Loss}_\alpha(u)=d_S (y, \gamma (w))+1$ (resp. $\mathsf{Loss}_\alpha(u)=d_S (y, \gamma (v))+1$, $\mathsf{Loss}_\alpha(u)=d_S (z, \gamma (w))+1$ and $\mathsf{Loss}_\alpha(u)=d_S (z, \gamma (v))+1$). If both $v$ and $w$ are mapped to $x$, $\mathsf{Loss}_\alpha(u)=0$, and if $v$ (resp. $w$) is mapped to $y$ or $z$ and $w$ (resp. $v$) is mapped to $y$ or $z$, it holds that $\mathsf{Loss}_\alpha(u)=d_S (y, \gamma (v))+d_S (z, \gamma (w))+2$ (resp. $\mathsf{Loss}_\alpha(u)=d_S (y, \gamma (v))+d_S (z, \gamma (w))+2$, $\mathsf{Loss}_\alpha(u)=d_S (y, \gamma (v))+d_S (v, \gamma (w))+2$, $\mathsf{Loss}_\alpha(u)=d_S (z, \gamma (v))+d_S (z, \gamma (w))+2$, $\mathsf{Loss}_\alpha(u)=d_S (z, \gamma (v))+d_S (z, \gamma (w))+2)$. Thus, by the inductive hypothesis, we have that {\em (i)} $\mathsf{c}(v,x)$ (resp. $\mathsf{c}(w,x)$) consists of the scores of the $k$-best DLT scenarios mapping the subtree of $G$ rooted in $v$ (resp. $w$) to the subtree of $S$ rooted in $x$, and {\em (ii)} $\mathsf{subtreeLoss}(v,y)$ (resp. $\mathsf{subtreeLoss}(w,z)$, $\mathsf{subtree}(w,y)$ and $\mathsf{subtreeLoss}(v,z)$) consists of the scores of the $k$-best DLT scenarios mapping the subtree of $G$ rooted in $v$ (resp. $w,w$ and $v$) to the subtree of $S$ rooted in some descendant of $y$ (resp. $z,y$ and $z$) with additional loss cost for each edge in the path from $y$ to $\gamma(v)$ (resp. $\gamma(w),\ \gamma(w)$ and $\gamma(v)$). Notice that a DLT scenario maps the subtree of $G$ rooted in $u$ to a subtree of $S$ rooted in $x$ under the constraint that the event corresponding to the matching of $u$ and $x$ is duplication if and only if it is a DLT scenario that matches $u$ and $x$, maps the subtree of $G$ rooted in $v$ to a subtree of $S$ rooted in a descendant of $x$ (which can be $x$ itself), and the subtree of $G$ rooted in $w$ to a subtree of $S$ rooted in a descendant of $x$ (which can be $x$ itself). Thus, it follows that $p_\Delta(u,x)$ is computed correctly.

Lastly, consider the list $p_\Theta(u,x)$. If $x$ is the root of $S$, then there does not exist a DLT scenario mapping the subtree of $G$ rooted in $u$ to the subtree of $S$ rooted in $x$ under the constraint that the event corresponding to the matching of $u$ and $x$ is horizontal transfer (because there is no vertex incomparable to $x$ to whom one of the children of $u$ should be mapped), and hence it is correct that each element $p_\Theta(u,x,i)$ remains with the assignment of $\infty$ as it was created. Now, suppose that $x$ is not the root of $S$. Then, by the pseudocode, $p_\Theta(u,x)$ consists of the $k$-best scores present in the following multisets: 
\begin{itemize}
\item $\{\mathsf{subtreeLoss}(v,x,j)+\mathsf{incomp}(w,x,r)~|~1\leq j,r\leq k\}$ and
\item $\{\mathsf{subtreeLoss}(w,x,j)+\mathsf{incomp}(v,x,r)~|~1\leq j,r\leq k\}$. 
\end{itemize}
Observe that the lists $\mathsf{subtreeLoss}(v,x),\mathsf{subtreeLoss}(w,x)$, $\mathsf{incomp}(w,x)$ and $\mathsf{incomp}(v,x)$ have already been computed. By the definition of $\mathsf{Loss}_\alpha(u)$ when $(u,w)\in \Xi$ , $\mathsf{Loss}_\alpha(u)=d_S (x, \gamma (v))$. Thus, by the inductive hypothesis, we have that {\em (i)} $\mathsf{subtreeLoss}(v,x)$ (resp. $\mathsf{subtreeLoss}(w,x)$) consists of the scores of the $k$-best DLT scenarios mapping the subtree of $G$ rooted in $v$ (resp. $w$) to a subtree of $S$ rooted in some descendant of $x$ (which can be $x$ itself), with additional loss cost for each edge in the path from $x$ to $\gamma(v)$ (resp. $\gamma(w)$). {\em (ii)} $\mathsf{incomp}(v,x)$ (resp. $\mathsf{incomp}(w,x)$) consists of the scores of the $k$-best DLT scenarios mapping the subtree of $G$ rooted in $v$ (resp. $w$) to a subtree of $S$ rooted in some vertex incomparable to $x$. Notice that a DLT scenario maps the subtree of $G$ rooted in $u$ to a subtree of $S$ rooted in $x$ under the constraint that the event corresponding to the matching of $u$ and $x$ is horizontal transfer if and only if it is a DLT scenario that matches $u$ and $x$, maps the subtree of $G$ rooted in one of the children of $u$ ($v$ or $w$) to a subtree of $S$ rooted in a descendant of of $x$ (which can be $x$ itself), and the subtree of $G$ rooted in the other child of $u$ to a subtree of $S$ rooted in a vertex incomparable to $x$.   Thus, it follows that $p_\Sigma(u,x)$ is computed correctly.
\end{proof}

\begin{observation}\label{obs:effRuntime}
Given an instance $(G,S,\sigma)$ of the DLT problem and a positive integer $k$, the algorithm runs in time $O(m\cdot n\cdot k)$ an requires $O(m\cdot n\cdot k)$ space.\end{observation}

\begin{proof}[Proof of observation \ref{obs:effRuntime}.]
For each pair of vertices $u\in V(G)$ and $x\in V(S)$, we construct a tuple of lists $(p_\Sigma(u,x),p_\Delta(u,x),p_\Theta(u,x),\mathsf{subtreeLoss}(u,x),$ $\mathsf{subtree}(u,x),\mathsf{incomp}(u,x))$. From the pseudocode, it is clear that the computation of each one of these lists is done in time $O(k)$. Thus, we have that the total running time is $O(m\cdot n\cdot k)$. As space is bounded by time, the observation follows.
\end{proof}

\subsection{Stage 2: Assigning Probabilities} \label{sec:probAssignment)}

In the second stage, we assign a probability to each hypernode in $\mathcal{H}$, so that a hypernode with best score has the highest probability, and hypernodes with score $\infty$ (the worst  possible score) have probability 0.

\subparagraph*{Weight Computation.}  Let $\gamma\in \mathbb{R}^+$ be a user-specified parameter. $\gamma$ is used to control the range between poorly scoring nodes versus top scoring nodes. As $\gamma$ grows lower, hypernodes with higher (worse) scores are assigned probabilities much lower than hypernodes with lower scores.

Denote $r=root$, and let $\mathsf{m}(r)$ be the largest integer $i\in\{1,\ldots,k\}$ such that $\mathsf{c}(r,i)\neq\infty$. ( Recall that the notation $(root,i)$ was defined in Section \ref{sec:hypergraph-def}). For a node $(r,i)$ where $i\in\{1,\ldots,\mathsf{m}(r)\}$, define
 $\displaystyle{\mathsf{w'}(r,i) =e^{\gamma\frac{\mathsf{c}(r,1)-\mathsf{c}(r,i)}{\mathsf{c}(r,1)-\mathsf{c}(r,\mathsf{m}(r))}}}$.
Then, the weight of a node $(r,i)$, which stands for the (unconditional) probability that the scenario described by $(r,i)$ happens, is defined as follows: if $i\in\{1,\ldots,\mathsf{m}(r)\}$, then $\mathsf{w}(r,i) = \frac{\mathsf{w'}(r,i)}{\sum_{j=1}^{\mathsf{m}(r)}{\mathsf{w'}(r,j)}}$; otherwise (i.e.~if $i\in\{\mathsf{m}(r)+1,\mathsf{m}(r)+2,\ldots,k\}$),~$\mathsf{w}(r,i) = 0$.

We now turn to define the weight of a hypernode $(u,x,i)$, which should stand for the (unconditional) probability that the scenario described by $(u,x,i)$ happens. 
The definition is recursive. In the basis, where $u$ is the root of $G$, we define $\mathsf{w}(u,x,i)$ (for any $x\in V(S)$ and $i\in \{1,\ldots,k\}$) as follows: if there exists an index $j\in\{1,\ldots,k\}$ such that $(r,j)$ is derived from $(u,x,i)$ (here, it means that they represent the same scenario), then $\mathsf{w}(u,x,i)=\mathsf{w}(r,j)$; otherwise, $\mathsf{w}(u,x,i)=0$.

Now, consider $v$ that is not the root of $G$. We define $\mathsf{w}(v,y,i)$ (for any $y\in V(S)$ and $i\in \{1,\ldots,k\}$) as follows.
First, let $\mathsf{D}(v,y,i)$ denote the collection of nodes $(u,x,j)$ such that $\mathsf{c}(u,x,j)$ was derived from $\mathsf{c}(v,y,i)$---in other words, the hypergraph has an hyperedge directed from $(v,y,i)$ (and some other node) to $(u,x,j)$. In particular, $u$ is the parent of $v$ in $G$, hence the weight $\mathsf{w}(u,x,j)$ is calculated before the weight $\mathsf{w}(v,y,i)$.
Then, define
$\mathsf{w}(v,y,i) = \sum_{(u,x,j)\in \mathsf{D}(v,y,i)}\mathsf{w}(u,x,j)$.

Note that $\sum_{i\in\{1,\ldots,k\}}\mathsf{w}(r,i) = 1$. 

\begin{lemma} \label{lem:prob} For any two compatible $u\in L(G)$ and $x\in L(S)$, $\mathsf{w}(u,x,1)=1$.
\end{lemma}

\begin{proof}[Proof of Lemma \ref{lem:prob}.]
We will verify a stronger property than the one in the statement of the lemma: For every vertex $u$ in the Gene tree $G$, it holds that
\[\sum_{x,i: (u,x,i)\in V({\cal H})}\mathsf{w}(u,x,i) = 1.\]
Before we verify this property, observe that when $u$ is a leaf, then $\mathsf{c}(u,x,1)=0$ for the unique vertex $x$ that is compatible with $u$, and $\mathsf{c}(u,x,i)=\infty$ (which means that $\mathsf{D}(u,x,i)=\emptyset$ and hence $\mathsf{w}(u,x,i)=0$) for any other pair $(x,i)$. Thus, the stronger property implies the correctness of the weaker statement regarding leaves.

To prove the (stronger) property above, we use induction. In the basis, $u$ is the root of the Gene tree $G$. Then, we have that $\sum_{x,i: (u,x,i)\in V({\cal H})}\mathsf{w}(u,x,i) = \sum_{i\in\{1,2,\ldots,k\}}\mathsf{w}(r,i) = 1$, and therefore the property holds. Now, suppose that $u$ is not the root of $G$, and that the property holds for each of its ancestors. Let $v$ be the parent of $u$ in $\cal H$. Then, we have that
\[\begin{array}{ll}
\displaystyle{\sum_{x,i: (u,x,i)\in V({\cal H})}\mathsf{w}(u,x,i)} & = \displaystyle{\sum_{x,i: (u,x,i)\in V({\cal H})}\sum_{y,j: (v,y,j)\in \mathsf{D}(u,x,i)}\mathsf{w}(v,y,j)}\\
& = \displaystyle{\sum_{y,i: (v,y,i)\in V({\cal H})}\mathsf{w}(v,y,i)} =1.
\end{array}\]
Here, the first equality follows directly from the definition of weights. The second equality follows from the fact that each hypernode $(v,y,i)$ (for any $y$ and $i$) that has positive weight is derived from exactly one hypernode $(u,x,j)$ (for some specific $x$ and $j$). (However, each hypernode $(u,x,j)$ can be used to derive several hypernodes $(v,y,i)$.) The last equality follows from the inductive hypothesis. This completes the proof.
\end{proof}

\begin{observation} Time and Space Complexity: Iterating the hypergraph in $O(|V(\mathcal{H})|)=O(m\cdot n\cdot k)$ time and space.
\end{observation}

\subsection{Stage 3.1: Pattern Discovery} \label{sec:pattern}


The current version of RSAM-finder allows pattern queries to be specified as follows. A pattern specification consists of a tuple $( \mathsf{EV},\mathsf{color},\mathsf{distance})$ where:
\begin{enumerate}
\item{$\mathsf{EV}\subseteq \{ \mathsf{S}, \mathsf{D},\mathsf{HT} \}$ specifies the evolutionary event of the pattern ($\mathsf{S}$ for speciation, $\mathsf{D}$ for duplication and $\mathsf{HT}$ for horizontal transfer).}

\item{$\mathsf{color}\in \{ \mathsf{red},\mathsf{black},\mathsf{None} \}$ specifies a color representing the environmental niche to which the sought RSAM confers adaptation.}
\item{$\mathsf{distance}\in \{\mathsf{True},\mathsf{False}\}$ is a boolean indicator specifying whether or not to consider edge lengths (representing evolutionary distances) in the pattern specification.} 
\end{enumerate}

For a colored query (having the second parameter in the specification set to $\mathsf{red}$ or $\mathsf{back}$), the user can provide, as part of the input, a function $\mathsf{colors}:L(X)\to \Upsilon$ where $X$ specifies whether the pattern refers to a subtree of $S$ or a subtree of $G$, and $\Upsilon=\{ \mathsf{red},\mathsf{black}\}$. Here, colors represent a binary environmental annotation of the leaves. Then, a preprocessing step is applied, in which the vertices of $S$ and $G$ are colored based on the colors assigned to the leaves of the subtree they root. We omit the technical details entailing the implementation of this preprocessing step to Section~1 of \citet{supmat2019}. 

In addition to the settings described above, the user can select one of two modes:
\begin{enumerate}
\item{{\bf Single-pattern mode.}  In this mode, the user specifies a single pattern and a threshold, and the sought RSAMs are identified as nodes $u \in I(G)$ such that $G_u$ is enriched in the pattern, and $|V(G_u)|$ is bounded from below by the specified threshold. }

\item{{\bf Dual-pattern (contrasting) mode.} In this mode, the user specifies two patterns and one threshold, and the sought RSAMs are identified as nodes $u\in I(G)$ with children $v,w\in V(G)$ such that  $G_v$  is enriched with one pattern while $G_w$ is enriched with the other pattern. Here, the subtree size bound threshold refers to $|V(G_v)|$ and $|V(G_w)|$.}
\end{enumerate}

The Pattern Identification algorithm proceeds as follows.
\begin{enumerate}
\item {For each pattern $P=(\mathsf{EV},\mathsf{color},\mathsf{distance})$ and for each hypernode $(u,x,i)\in V(\mathcal{H})$, check whether both $\mathsf{event}(u,x,i)\in \mathsf{EV}$ and the colors obey the requirements derived from the $\mathsf{color}$ field of the pattern specification (described in more details in  Section~1 of  \citet{supmat2019}). 
If so, mark $(u,x,i)$ as interesting. \label{subsec:hypergraph_identification}}

\item {Reflect the interesting nodes identified in $\mathcal{H}$ to $G$, by assigning corresponding weights to $V(G)$. Each $u \in I(G)$ is assigned a score, which is the sum of the probabilities of instances of the pattern found in $G_u$, normalized by the number of possible events in $G_u$. Additional book-keeping details regarding how this score is computed are given in Section~\ref{sec:calculate_scores}.

Based on the specified mode of the query (single pattern or dual pattern), identify the $t$ top scoring vertices $u\in I(G)$. In case of a single-pattern mode, the scores are as defined in (2). 
In case of  dual-pattern mode, let $P_1$ and $P_2$ be the patterns. For each $u\in V(G)$ with children $v,w\in V(G)$ the score of $u$ is score of $v$ for $P_1$ (as defined in (2)) plus the score of $w$ for $P_2$, and vice versa (that is, each vertex is assigned two scores).  \label{subsec:scoring}}

\end{enumerate}
%
%
%
%

\begin{observation} Time and Space Complexity: Iterating over the hypergraph takes $O(|V(\mathcal{H})|)=O(m\cdot n\cdot k)$ time and space. 
\end{observation}

\subsection{Stage 3.2: Score Computation.} \label{sec:calculate_scores}
For each defined pattern $\mathsf{P}=(\mathsf{EV},\mathsf{color},\mathsf{distance})$, let $\mathsf{counter}_P:V(G)\to \mathbb{R}^+$ be a counter, initialized by $0$. For each $u\in I(G)$ ,let $v$ and $w$ be its right and left children, respectively. Let \[\mathcal{I}_u=\{ (u,x,i)\in V(\mathcal{H}) :(u,x,i) \text{ is marked as interesting with respect to }P \}.\] That is, for each vertex $x\in V(S)$ and $i\in \{1,\dots ,k \}$ such that $(u,x,i)\in V(\mathcal{H})$ was marked as interesting in stage 1 with respect to pattern $P$,  $(u,x,i)\in \mathcal{I}_u$. Let \begin{gather*}
\mathsf{counter}_P(u) =\mathsf{counter}_P(v) + \mathsf{counter}_P(w) + \sum_{(u,x,i)\in \mathcal{I}_u}\mathsf{w}(u,x,i)
\end{gather*}
where $ \mathsf{w}(u,x,i)$ are the probabilities assigned in Section \ref{sec:probAssignment)}.
Intuitively, for each vertex $u\in V(G)$ we calculate its probability to be interesting, with respect to the patterns we defined.
In order to avoid a bias due to variation in the sizes of the subtrees rooted by the competitively estimated nodes in $G$, we normalise each value by the number of edges in the subtree rooted in the vertex times $k$, which is an upper bound on the number of possible patterns in all the solutions. That is, for each vertex $u\in V(G)$ and pattern $P$, let $\mathsf{counter}_P(u)=\frac {\mathsf{counter}_P(u)} {\mid E(G_u)\mid \cdot k}$.

\section{Experimental Results}

We implemented the algorithm described in this paper as a tool, denoted RSAM-finder, and made it publicly available via GitHub (\cite{zoller2019}).

In this section we test and exemplify the performance of RSAM-finder. The tests are based on large scale simulations, where we demonstrate the engine's tolerance to noise (Subsection \ref{sec:simulations_noise}), and measure the practical running times of the proposed hypergraph construction algorithm as a function of increasing input size (Subsection \ref{sec:running_time}). In Subsection \ref{sec:BetaLactamase} we exemplify an application of our proposed approach to the discovery and analysis of RSAMs in a Beta Lactamase gene. But first, in Subsection \ref{sec:methods}, we give the technical details regarding our simulations, tests and experiments.

\subsection{Methods and Data Bases}\label{sec:methods}
Genes in our experiment are represented by their membership in Cluster of Orthologous Genes (\citet{tatusov2000cog}).  The STRING database (\citet{szklarczyk2016string}) was used to extract the chromosomal protein sequences for the COGs of interest, annotated with their corresponding species names as well as the corresponding NCBI IDs.  
Protein sequences were subjected to multiple sequence alignment and dendogram construction via Clustal Omega (\citet{sievers2018clustal}).  The list of NCBI IDs was used as input for \textit{NCBI Taxamony Browser} which provided a (non-binary) Species tree.
Both Gene and Species trees were converted to binary trees via the Ape R package (\citet{popescu2012ape}).
Habitat labels for the species were extracted from \href{https://www.patricbrc.org}{PATRIC}, and missing tags were manually annotated by information from the GOLD database (\citet{mukherjee2016genomes}) and from literature. CD Search (\citet{marchler2004cd}) was employed to seek statistically significant discriminating domain-level mutations (i.e.~the gain or loss of a protein functional domain). 
The simulator and our algorithm  were implemented in Python, using \href{https://networkx.github.io}{NetworkX package}, DendroPy (\citet{sukumaran2010dendropy}) and ETE Toolkit (\citet{huerta2016ete}). Visualization of the trees and plots were created using Matpllotlib and Seaborn tools.

For the simulation-based experiments, we generated random binary trees. The generation of a random binary tree was done in a top down manner, using the ETE Toolkit~(\citet{huerta2016ete}). We began with a given set of vertices, based on which we created a random binary tree. The tree was duplicated and one copy was denoted $G$, while the other was denoted $S$. The function $\sigma:L(G)\to L(S)$ was implemented as the matching between each leaf in $G$ to its copy in $S$, and the function $\mathsf{color}:L(S)\to \{ \mathsf{red},\mathsf{black} \}$ was implemented as a random binary function.  To implant the pattern in the resulting random trees, we picked a random vertex $u\in V(G)$, and modified the function $\sigma:L(G)\to L(S)$ for all vertices $w\in L(G_u)$ in a way that created a Horizontal Transfer event. To this end, consider a vertex $w\in V(G_u)$. Vertex $w$ is made to represent a Horizontal Transfer event as follows. Let $x\in V(S)$ be the copy of $w$ in $S$. Let $L(S_x)$ denote the copy of $L(G_w)$ (the leaves of the subtree rooted in $w$) in the Species tree, thus the function $\sigma$ maps each leaf of $G_w$ to its copy in the leaves of $S_x$. Then, to create a Horizontal Transfer in $w$, we need to find a vertex $y\in V(S)$ such that $y$ and $x$ are incomparable, and change the mapping of the leaves of $G_w$ to the leaves of $S_y$ randomly -- that is, for each vertex $r\in G_w$ define $\sigma(r)$ to be a random vertex $z\in S_y$. This is likely to create a Horizontal Transfer in the DLT-reconciliation. Recall that in addition, we want to make those planted  Horizontal Transfer events red-to-red events. To achieve this, we check to see if the random vertices $u$ and $y$, which are the source and the target of the Horizontal Transfer, are ``mostly red", as defined in  Section~1 of \citet{supmat2019}. If they are not, we make another random choice and check the colors again. The query pattern $(\{ \mathsf{HT}\} ,\mathsf{red},\mathsf{True})$ was used in the simulation-based experiments. According to this pattern, we sought subtrees that are enriched in red-to-red Horizontal Transfer events. (For additional details, see Section \ref{sec:pattern}.)


\begin{figure}[t]
\begin{center}
\fbox{\includegraphics[scale=0.1]{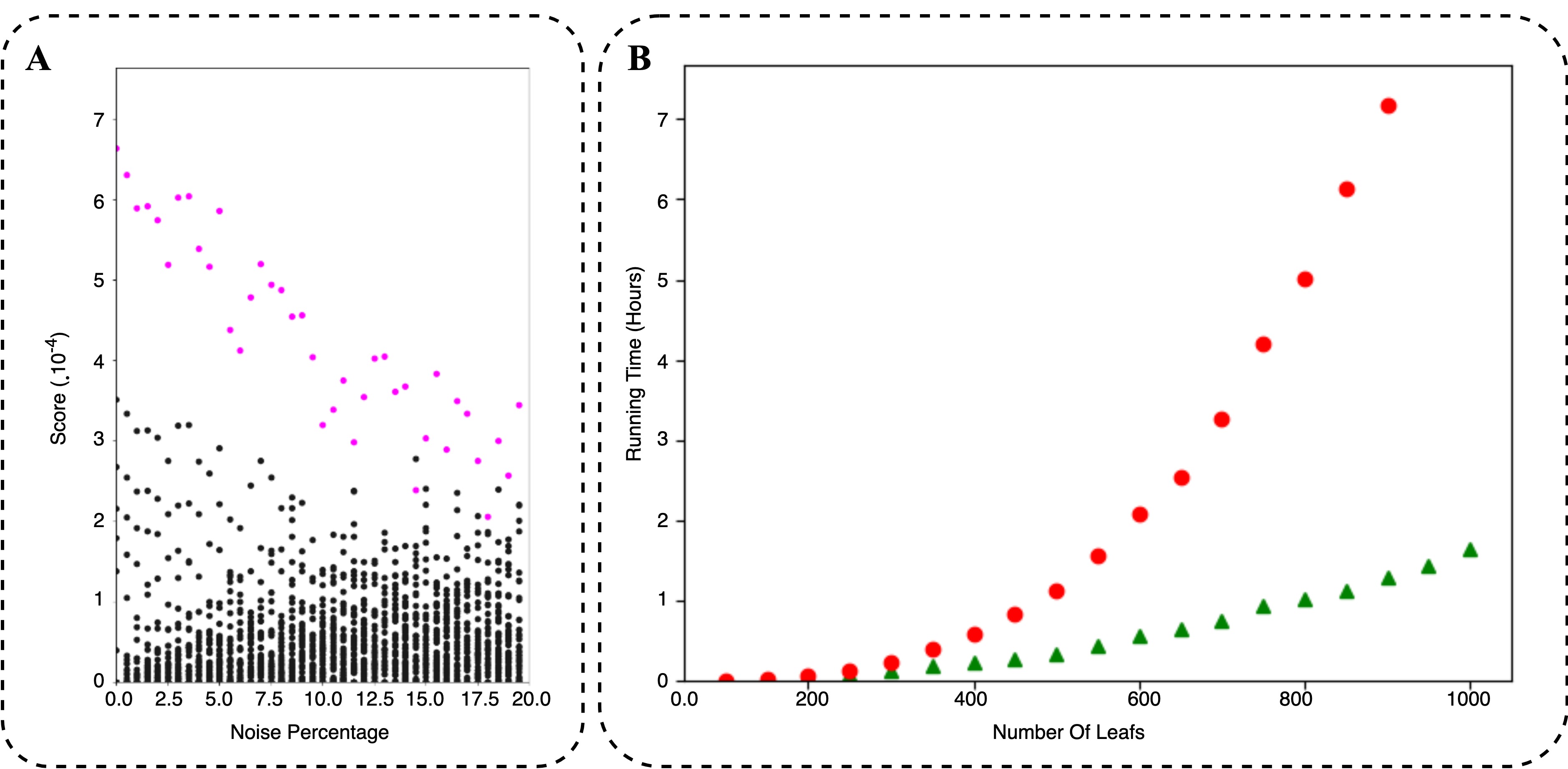}}
\end{center}
\caption{(A) The scores of the vertices in different noise levels on the input. The purple dots represent the planted vertex, and they obey the sought pattern. (B) Running times of the naive and efficient algorithms. Green triangles represents the efficient version and red circles represents the running times of the naive algorithm.}
\label{fig:simulations}
\end{figure}

\subsection{Testing for Noise Tolerance}\label{sec:simulations_noise}
We tested our tool on a random data set that was generated as described above, by introducing into the simulations an additional ``noise factor" affecting Horizontal Transfers and colors.
Each noise level represents the level of random changes in $\sigma$ and random colors of the species. In particular, a noise level of 0\% means that no changes were done to the mapping between the leaves of the Gene tree to the leaves of the Species tree except those of the planted pattern, and no change was made to the function  $\mathsf{color}:L(S)\to \{ \mathsf{red},\mathsf{black} \}$, while a noise level of 100\% means that the mappings of all of the vertices of the Gene tree were randomly picked, and all of the species colors were randomly picked again.

Fig.~\ref{fig:simulations}(A) demonstrates the advantage of our approach across different noise levels, following the strategy described above to generate randomized phylogenetic Gene and Species trees with a planted pattern. First, we constructed random phylogenetic Species and Gene trees with 600 leaves and one planted pattern (marked in purple in Fig.~\ref{fig:simulations}(A)). For each noise level between 0\% to 20\% we constructed the corresponding hypergraph. For all experiments, we used $k=100$, set the minimum size of a subtree to 0.1\% of the number of all edges, and set $c_\Delta=c_\Theta=1$. Results for each noise level were computed as an average of 50 random choices for the same noise level, on the same input trees. The scores are as defined in Section \ref{sec:pattern}. 

We found that, at the lower noise levels, the score of the planted vertex $u\in V(G)$ is higher than that of any other vertex, and this difference decreases as the noise level increases. Note that the additional noise increases the number and scores of false positives found. These findings support the claim that our method is able to find a pattern within noisy data. 

\subsection{Running Time Measurements}\label{sec:running_time}

To demonstrate the practicality of the theoretical improvements presented in Section \ref{sec:hypergraphCons}, we compared the running times of  the efficient, $O(m n k )$ time algorithm for hypergraph construction proposed in Section \ref{sec:hypergraphCons}, versus the naive, $O(m n (n+k) log(n+k))$ time algorithm mentioned in the introduction.  

The inputs to the compared algorithms were generated as follows: We picked random binary trees denoted $S$ and $G$, with number of leaves ranging from $100$ to $1000$. For each number of leaves, we randomly created 10 such pairs of trees, and ran both the naive and the efficient algorithms on both datasets. 

Fig.~\ref{fig:simulations}(B) summarizes the measured time results. The green triangles correspond to the average of the time measured for the efficient version of the algorithm, and the red circles correspond to an average of the time measured for the naive version of the algorithm. 

We found that as we increase the number of leaves, the differences in practical running times between the naive and the efficient algorithms become more significant. Furthermore, as expected in practice, the running time of the efficient algorithm is linear in the input size, while that of the naive one behaves as a non-linear function.

\subsection{Example: RSAM Discovery in a Beta Lactamase.}\label{sec:BetaLactamase}

 \begin{figure}[!htb]
\begin{center}
\fbox{\includegraphics[scale=0.1]{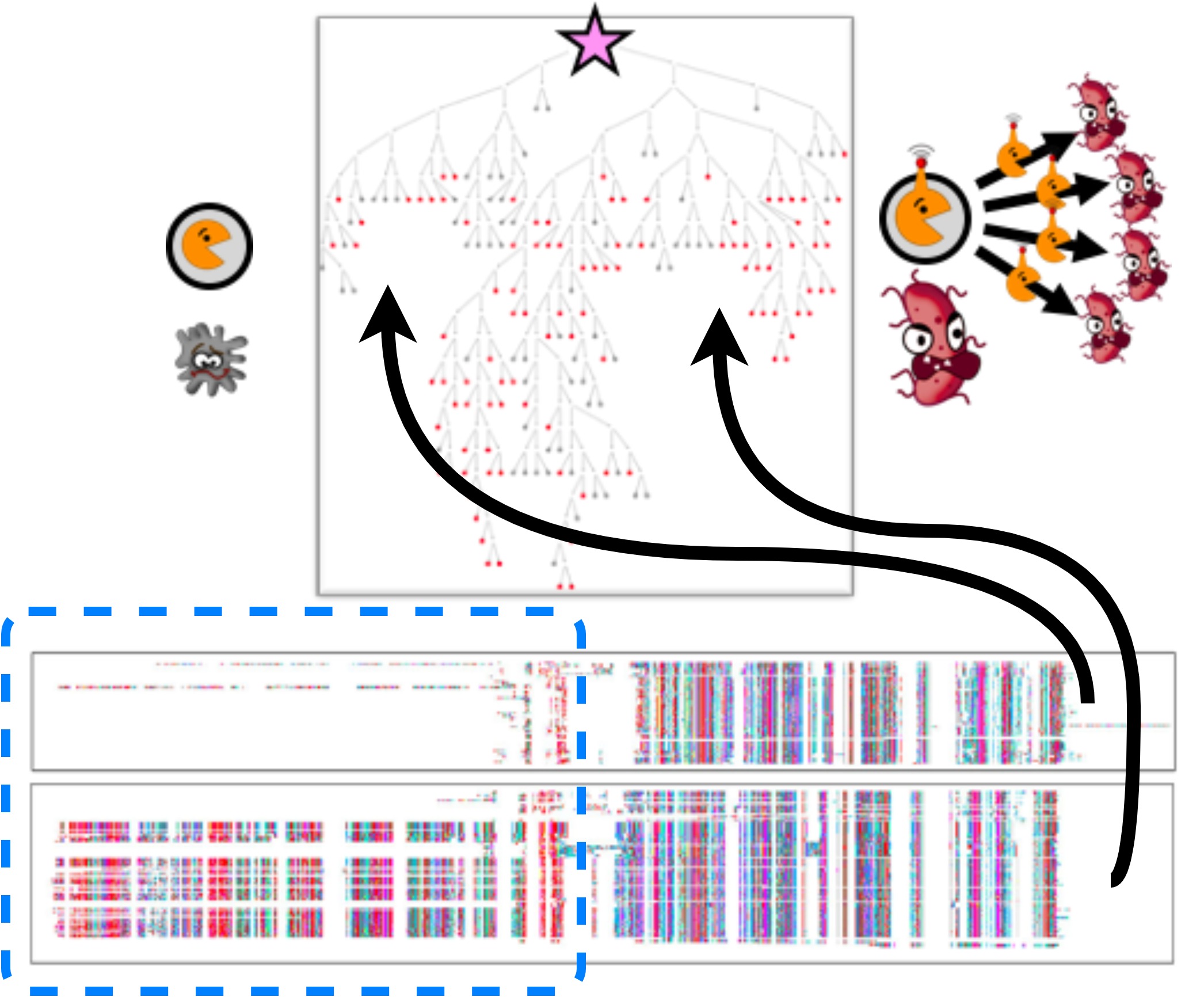}}
\end{center}

\caption{Application of RSAM-finder to genes belonging to the class D Beta Lactamase family. The sought pattern is $ \left( (\{ \mathsf{HT}\} ,\mathsf{red},\mathsf{True}),( \{\mathsf{S},\mathsf{D},\mathsf{HT}\} ,\mathsf{black},\mathsf{False})\right)$, which codes for two patterns, one of a massive Horizontal Transfer events from red to red (right subtree) and the other is all black events (left subtree). The figure shows the top-scoring subtree, and the corresponding sequences. The blue rectangle marks the mutation characterizing  the sequences in the leaves of the right subtree: this insertion was identified as a BlaR (signal transducer) domain.}
\label{fig:class_d_result}
\end{figure}
 
Beta lactamases are versatile enzymes conferring resistance to the Beta lactam antibiotics,
found in a diversity of bacterial sources. Their commonality is the ability to hydrolyze
chemical compounds containing a Beta lactam ring  (\citet{bush2018past}). The secretion of antimicrobial compounds is an ancient mechanism with clear survival benefits for microbes competing with other microorganisms. Consequently, mechanisms that confer resistance are also ancient and may represent an underestimated reservoir in environmental bacteria (\citet{bush2018past}). Antibiotic resistance factors, conferring adaptation to the pathogenesis environment, are widely spread by horizontal gene transfer mechanisms like conjugation, transformation and transduction (\citet{navarro2006acquisition,poirel2009integron}). The persistent exposure of bacterial strains to a multitude of Beta lactams has induced dynamic and continuous production and mutation of Beta lactamases in these bacteria, expanding their activity even against the newly developed Beta lactam antibiotics (\citet{stapleton2016outbreaks}).  Thus, an important objective is to identify mutations in Beta lactamase genes conferring adaptation to human and animal hosts.


Motivated by the above, we exemplify a microbiological application of RSAM-finder to the discovery of RSAMs in Beta Lactamase genes that confer adaptation to human and animal hosts. To this end, we use the pattern $ \left( (\{ \mathsf{HT}\} ,\mathsf{red},\mathsf{True}),( \{\mathsf{S},\mathsf{D},\mathsf{HT}\} ,\mathsf{black},\mathsf{False})\right)$ to the discovery of RSAMs in Beta Lactamase genes that confer adaptation to human and animal hosts. Here, colors represent a binary environmental annotation: human and animal host (219 species) were annotated ``red'', while species associated with all other habitats (324 species), such as soil, water and plant, were annotated ``black''.

Among the known classes (A-D) of Beta lactamase, class D (represented by COG2602) is considered to be the most diverse (\citet{evans2014oxa}). Thus, we selected COG2602 (622 genes in 543 genomes) as the dataset for our example.  Parameters were set as follows: $k=50$, the minimum size required per sought subtree was set to $0.1$ of the total number leaves of $G$, $c_\Delta=c_\Theta=1$ and $c_\Sigma=c_\mathsf{Loss}=0$. A figure displaying $G$, where the top-ranking RSAM node is marked with a star, is given in the supplementary materials. Also provided are the corresponding sequences, a figure displaying the corresponding $S$, and $\sigma$.

Within the top-ranking result for this query, we were interested in the subtree matching the first part of pattern (i.e.~enrichment in red-to-red HT edges). The gene set represented by the leaves of  this subtree (denoted ``identified gene set'') was found to be enriched in an additional domain, BlaR, a signal transducer membrane protein regulating Beta lactamase production (87/119 in the identified gene set  versus 118/622 in the background, p-val =  3.94e-52). The only transcriptional regulator currently known for Beta lactamase genes is the repressor protein BlaI, previously predicted to operate in a two-component regulatory system together with BlaR in Class A Beta lactamase (\citet{alksne1997expression}).  The positions adjacent to the instances of the identified gene set in the corresponding genomes were found to be  enriched in BlaI (70/119 of the identified gene set instances versus 90/622 of the background gene set instances, hypergeometric p-value = 1.11e-41). Note that this result was obtained with $c_\mathsf{Loss}$ set to 0. When repeating the experiment with $c_\mathsf{Loss} = 1$, this result is still found among the the two top ranking vertices.

In contrast to the identified gene set, the genes represented by the subtree that matches the second part of the pattern (frequent black events of all types) are not enriched in the BlaR domain (2/36), nor is there contextual enrichment in BlaI (4/36) in positions immediately adjacent to instances of these genes. Applying RSAM-finder to this data with simpler queries that take into account only enrichment in environmental coloring does not yield this result, nor does the application of RSAM-finder to this data with any part of the pattern on its own. 
 

The identified gene set for this result spans a wide range of Firmicutes, including  both pathogenic (e.g.~staphylococcus) and non-pathogenic species (e.g.~various gut microbes from the Clostridiales order). Homology between BlaR receptor proteins and the extra-cellular domain of Class D Beta-lactamases was previously observed (\citet{massidda1996borderline, brandt2017silico}), mainly in {\em gram-negative bacteria} (with focus on clinical samples). 
%
Thus, RSAM-finder identifies a putative Beta lactamase system in {\em gram positive bacteria}, consisting of a {\em COG2602-BlaR} Beta lactamase-receptor protein and its  BlaI family repressor, {\em predicted to confer adaptation to animal and human host environment.}
Further comparative sequence-level analysis (\citet{toth2016class}) may reveal the affinity of this Beta lactamase system to specific Beta lactam~drugs.

 \label{sec:bio}

\section{Conclusions}

We defined a new optimization problem in the DLT reconciliation domain. The input to this problem consists of a gene tree, constructed for a given gene orthology group, a species tree constructed for the species harboring one or more members of this gene orthology group, and a pattern representing a sought scenario in the reconciliation of the two trees. The sought pattern could imply some evolutionary process of interest, such as e.g. a gene conferring adaptation of the species to a specific environmental niche.  The goal of the problem is to compute, for any vertex in the gene tree,  a score reflecting the probability that  the genomic mutations associated with the edge leading into this vertex confer the occurrence of the sought pattern within high-scoring reconciliations of the subtree rooted by this vertex with corresponding subtrees in the species trees.

To solve this new problem, and overcome some of the noise associated with gene tree and species tree reconstruction, we proposed an algorithm that  first constructs a hypergraph $\cal H$ that stores information regarding the k-best DLT reconciliation scenarios for a given problem instance. 
The time complexity of the algorithm we propose for the construction of this hypergraph is $O(m\cdot n\cdot k)$, which is essentially {\em optimal} since the number of vertices (and hence also the size) of the hypergraph can be as large as $\Omega(m\cdot n\cdot k)$. 

Interesting open problems include the goal of extending the tool to handle more robust variations of phylogenies, such as polytomies and phylogenetic networks. It may also be helpful to consider bootstrapping methods to train the parameters and thresholds utilized by RSAM-Finder.

\section{Acknowledgments}

This work was supported by ISF grants no.~1176/18 and no.~939/18 and by the Lynn and William Frankel Center for Computer Science.

\section{Disclosure Statement}

No competing financial interests exist.

\bibliography{ms}

\end{document}